\RequirePackage{fix-cm}
\documentclass[smallextended, natbib]{svjour3}       

\usepackage{soul}
\usepackage{amsmath}
\usepackage{amsfonts}
\usepackage{amssymb}
\usepackage{mathtools}
\usepackage{mathrsfs}
\usepackage{nicefrac}
\usepackage[linesnumbered, ruled]{algorithm2e}

\smartqed  

\usepackage{rotating}
\usepackage{pdfpages}
\usepackage{graphicx}
\usepackage{marvosym}	
\usepackage{xcolor}
\usepackage{csquotes}
\usepackage{float}
\usepackage[latin1]{inputenc}
\usepackage{url}
\usepackage[normalem]{ulem}

\begin{document}

\title{Classes of treebased networks}
%


\author{Mareike Fischer         \and
		Michelle Galla 			\and
		Lina Herbst				\and
		Yangjing Long			\and
		Kristina Wicke
}


\institute{
	Mareike Fischer (\Letter) 	\email{email@mareikefischer.de}
            \\
     Mareike Fischer, Michelle Galla, Lina Herbst, Kristina Wicke \at Institute of Mathematics and Computer Science, University of Greifswald, Germany \\ 
     Yangjing Long \at School of Mathematics and Statistics, Central China Normal University, Wuhan, Hubei, China\\
}

\date{Received: date / Accepted: date}

\maketitle

\begin{abstract}
Recently, so-called treebased phylogenetic networks have gained considerable interest in the literature, where a treebased network is a network that can be constructed from a phylogenetic tree, called the base tree, by adding additional edges. 
The main aim of this manuscript is to provide some sufficient criteria for treebasedness by reducing phylogenetic networks to related graph structures. While it is generally known that deciding whether a network is treebased is NP-complete, one of these criteria, namely edgebasedness, can be verified in linear time. Surprisingly, the class of edgebased networks is closely related to a well-known family of graphs, namely the class of generalized series parallel graphs, and we will explore this relationship in full detail. Additionally, we introduce further classes of treebased networks and analyze their relationships.

\end{abstract}

\section{Introduction}
Phylogenetic networks are of considerable interest in the current literature as they allow for the representation of non-treelike evolutionary events, such as hybridization and horizontal gene transfer. 

Various classes of phylogenetic networks have been introduced and studied, one of them being the class of so-called treebased networks. Roughly speaking, a phylogenetic network is treebased if it can be obtained from a phylogenetic tree by adding additional edges. 

While \cite{Francis2015} first introduced this concept for binary rooted phylogenetic networks, recently \cite{Francis2018} extended it to binary unrooted networks, \cite{Jetten2018} to non-binary rooted networks and \cite{Hendriksen2018} and \cite{paper2} to non-binary unrooted networks. 

In the present manuscript, we focus on unrooted networks and consider both the binary and non-binary case. 
First, we introduce three procedures that reduce a phylogenetic network to related graphs. This leads to some sufficient criteria which guarantee a phylogenetic network to be treebased (whether they are binary or not). Some of these criteria are based on classical graph theory, particularly on the theory of Hamiltonian paths, cycles, and graphs. 
Another sufficient criterion for treebasedness is a property that we refer to as edgebasedness. This criterion is again related to classical graph theory, namely to so-called generalized series parallel graphs. We will introduce this concept in full detail, highlight the relationship between edgebased graphs and generalized series parallel graphs and analyze its implications. In particular, we remark that edgebasedness can be tested in linear time since generalized series parallel graphs can be recognized in linear time. This is also of practical relevance as, in general, the problem of deciding whether a network is treebased or not is an NP-complete problem (cf. \citet{Francis2018}).

The remainder of this manuscript is organized as follows. In Section \ref{sec_preliminaries}, we introduce some basic phylogenetic and graph-theoretical concepts and terminology. We then introduce three procedures, leaf cutting, leaf shrinking and leaf connecting that reduce a phylogenetic network to related graphs in Section \ref{NetworkReduction}. This leads to some sufficient criteria for treebasedness, e.g. edgebaseness, and some classes of phylogenetic networks that are guaranteed to be treebased, which we introduce in Section \ref{sec_classes}. After summarizing the relationships between these classes, we conclude this manuscript with Section \ref{sec_discussion}, where we discuss our results and indicate possible directions of future research.

\section{Methods} 
We use mathematical proof techniques to obtain our results based on the definitions and methods defined in the Preliminaries.

\subsection{Preliminaries}\label{sec_preliminaries}
\subsubsection*{Phylogenetic and basic graph-theoretical concepts}
Throughout this manuscript let $G=(V(G),E(G))$ (or $G=(V,E)$ for short) denote a graph with vertex set $V(G)$ and edge set $E(G)$.  
Note that throughout this manuscript graphs may contain parallel edges and loops. Whenever we require graphs without parallel edges and/or loops, we specifically speak of \emph{simple} graphs and whenever parallel edges are allowed but loops are not, we speak of \emph{loopless} graphs.
Furthermore, we use $N_G(v)$ (or $N(v)$ for short if there is no ambiguity) to denote the neighborhood of a vertex $v$ in $G$, i.e. the set of vertices adjacent to $v$ in $G$. Note that if $G$ is a simple graph without parallel edges and loops, we have $\vert N_G(v) \vert = deg(v)$.

Now, let $X$ denote a finite set (e.g. of taxa or species) with $\vert X \vert \geq 1$.
An \emph{unrooted phylogenetic network} $N^u$ (on $X$) is a connected, simple graph $G=(V,E)$ with $X \subseteq V$ and no vertices of degree 2, where the set of degree-1 vertices (referred to as the \emph{leaves} or \emph{taxa} of the network) is bijectively labeled by $X$. Such an unrooted network is called \emph{unrooted binary} if every inner vertex $u \in V \setminus X$ has degree 3. It is called a \emph{phylogenetic tree} if the underlying graph structure is a tree.
In the following, we denote by $\mathring{E}$ the set of inner edges of $N^u$, i.e. those edges that are {\em not} incident to a leaf. 
A phylogenetic network $N^u = (V,E)$ on $X$ is called \emph{treebased} if there is a spanning tree $T = (V,E')$ in $N^u$ (with $E' \subseteq E$) whose leaf set is equal to $X$. This spanning tree is then called a \emph{support tree} for $N^u$. Moreover, the tree $T'$ that can be obtained from $T$ by suppressing potential degree-2 vertices is called a \emph{base tree} for $N^u$. Note that the existence of a support tree $T$ for $N^u$ implies the existence of a base tree $T'$ for $N^u$.

When analyzing networks, or more generally connected graphs, it is often useful to decompose them into simpler pieces, which can then be analyzed individually. Therefore, let $G=(V,E)$ be a connected graph. A \emph{cut edge}, or \emph{bridge}, of $G$ is an edge $e$ whose removal disconnects the graph. Similarly, a vertex $v$ is a \emph{cut vertex} (sometimes also called an \emph{articulation}) if deleting $v$ and all its incident edges disconnects the graph. Moreover, a set $\mathcal{C}$ of vertices whose removal disconnects a graph is called a \emph{separating set} or \emph{vertex cut}.

If after the removal of a cut edge one of the induced connected components of the resulting graph is a single vertex, the corresponding cut edge is called \emph{trivial}.
We call $N^u$ a simple network if all of its cut edges are trivial.

A \emph{blob} in a connected graph (and more specifically in a network) is a maximal connected subgraph that has no cut edge. If a blob consists only of one vertex, we call the blob \emph{trivial}. Note, however, that a blob may contain cut vertices. An example for such a blob can be seen in Figure \ref{blob}. Moreover, note that we can consider a network as a \enquote{tree} with blobs as vertices (cf. \citet{Gusfield2005}). \label{tree_of_blobs}
In contrast, a \emph{block} in a connected graph $G$ is a maximal induced \emph{biconnected} subgraph of $G$, i.e. a maximal induced subgraph that remains connected if any one of its vertices is removed. In particular, a block does not contain cut vertices.

\begin{figure}[htbp] 
	\centering
	\includegraphics[scale=0.45]{Figure1.eps}
	\caption{Unrooted non-binary phylogenetic network $N^u$ on leaves $1$,$2$,$3$ and $4$. The gray areas correspond to the blobs of $N^u$. Notice that the biggest blob contains a cut vertex (depicted as a square vertex). Moreover, notice that $N^u$ can be considered as a tree with blobs as vertices, as the cut edges and blobs of $N^u$ induce a \enquote{tree structure}.}
	\label{blob}
\end{figure}

Following \citet{paper2}, we call a graph $G$ (or a network $N^u$) \emph{proper} if the removal of any cut edge or cut vertex present in the graph (or the network) leads to connected components containing at least one leaf each.


Lastly, two important operations on graphs that will be used in the following are the concepts of \emph{subdividing an edge} and \emph{suppressing a vertex}.
Therefore, let $G$ be a graph with some edge $e=\{u,v\}$. Then, we say that we \emph{subdivide} $e$ by deleting $e$, adding a new vertex $w$ and adding the edges $\{u,w\}$ and $\{w,v\}$. The new degree-2 vertex $w$ is sometimes also called an \emph{attachment point}. Note that we often also refer to the vertex incident to a vertex $x$ of degree $1$, i.e. incident to a leaf $x$, as the attachment point of $x$, even if it is a vertex of degree higher than three. On the opposite, given a degree-2 vertex $w$ with adjacent vertices $u$ and $v$, by \emph{suppressing $w$} we mean deleting $w$ and its two incident edges $\{u,w\}$ and $\{w,v\}$ and adding a new edge $\{u,v\}$.

\subsection*{Further graph-theoretical concepts} \label{graphtheo}
Before we can introduce three procedures to reduce a phylogenetic network to related graphs, we need to recall some basic concepts from classical graph theory. Most importantly, we need to recall the notion of Hamiltonian paths and Hamiltonian cycles.

A \emph{Hamiltonian path} is a path in a graph that visits each vertex exactly once. 
If this path is a cycle, we call the path a \emph{Hamiltonian cycle}. Moreover, a graph that contains a Hamiltonian cycle is called a \emph{Hamiltonian graph}. 
A graph is called \emph{Hamilton connected} if for every two vertices $u$, $v$ there is a Hamiltonian path from $u$ to $v$. Note that in particular, every Hamilton connected graph is Hamiltonian, because the strong property of Hamilton connectedness also holds for adjacent vertices, so the edge $e=\{u,v\}$ together with the Hamilton path from $u$ to $v$ forms a Hamiltonian cycle. 
As has been noted in \citet{Francis2018}, there is a strong connection between Hamiltonian paths and treebasedness of phylogenetic networks. However, before we can elaborate this in more detail, we need to introduce a few more concepts.

First, recall that the \emph{toughness} $t(G)$ of a graph $G$ (or, analogously, of a phylogenetic network $N$) is defined as $$t(G) = \min \limits_{\mathcal{C}} \frac{|\mathcal{C}|}{c(G-\mathcal{C})}, $$ where the minimum is taken over all separating sets $\mathcal{C}$ of $G$, $G-\mathcal{C}$ denotes the (disconnected) graph that results from deleting all vertices of $\mathcal{C}$ from $G$ and all edges incident to $\mathcal{C}$ and where $c(G - \mathcal{C})$ denotes the number of connected components in $G - \mathcal{C}$. The concept of toughness plays an important role in the study of Hamiltonian graphs \citep{Chvatal1973,Kabela2017}, and thus, as we will show, also for treebasedness of a network.  

Next, another graph theoretic concept we will consider are \emph{chordal} graphs. Recall that a graph is called chordal if all cycles of length four or more have a chord, i.e. an edge that connects two vertices of the cycle which are not adjacent in the cycle \citep[p. 135]{Diestel2017}. We will call a phylogenetic network chordal if its underlying graph is chordal. 

Finally, recall that in graph theory, if a graph $G$ can be converted into another graph $G'$ by a sequence of vertex deletions, edge deletions and suppression of degree-2 vertices, $G'$ is called a \emph{topological subgraph} of $G$ (cf. \citet{Grohe2011}). In the present manuscript, we will consider a restricted version of topological subgraphs. In particular, we call a graph $G'$ a \emph{restricted topological subgraph} of a graph $G$, if $G$ can be converted into $G'$ by a sequence of the following operations:

\begin{enumerate} \label{restriction_operations}
\item Delete a leaf (and its incident edge).
\item Suppress a vertex of degree 2.
\item Delete one copy of a multiple edge, i.e. if $e_1=e_2 \in E(G)$, delete $e_2$.
\item Delete a loop, i.e. if $e=\{u,u\} \in E(G)$, delete $e$. 
\end{enumerate}

Note that $G'$ is in this case also a topological subgraph as the above operations are restricted versions of the respective operations which lead to topological subgraphs: deleting a leaf is a special kind of vertex deletion, and the deletion of a multiple edge or of a loop are special types of edge deletions. \\

Last but not least, a connected and loopless graph $G$ is called a \emph{generalized series parallel graph} ($GSP$ graph for short) if it can be reduced to a single edge, i.e. to the complete graph $K_2$, by only applying operations 1. -- 3., i.e. by only deleting leaves, suppressing degree-2 vertices or deleting parallel edges (see for example \citet{Ho1999}). \label{gsp}
Similarly, a connected and loopless graph $G$ is called a \emph{series parallel graph} ($SP$ graph for short) if it can be reduced to $K_2$ by operations 2 and 3, i.e. by suppressing degree-2 vertices or deleting parallel edges (see for example \citet{Ho1999}). \label{sp}

Both $GSP$ and $SP$ graphs belong to the class of $2$-terminal graphs as shown by the following definition:
\begin{definition}[adapted from \citet{Ho1999}] \label{Def_GSP}
\begin{enumerate}
\item The graph $K_2$ consisting of two vertices $u$ and $v$ (called \emph{terminals}) and a single edge $\{u,v\}$ is a \emph{primitive} $GSP$ graph. 
\item If $G_1$ and $G_2$ are two $GSP$ graphs with terminals $u_1, v_1$ and $u_2, v_2$, respectively, then the graph obtained by means of any of the following three operations is a $GSP$ graph:
	\begin{enumerate}
	\item Series composition of $G_1$ and $G_2$: identifying $v_1$ with $u_2$ and specifying $u_1$ and $v_2$ as the terminals of the resulting graph.
	\item The parallel composition of $G_1$ and $G_2$: identifying $u_1$ with $u_2$ and $v_1$ with $v_2$, and specifying $u_1$ and $v_1$ as the terminals of the resulting graph. 
	\item The generalized-series composition of $G_1$ and $G_2$: identifying $v_1$ with $u_2$ and specifying $u_2$ and $v_2$ as the terminals of the resulting graph.
	\end{enumerate}
\end{enumerate}
\end{definition}

Now, the family of $SP$ graphs consists of those $GSP$ graphs that are obtained by using only the series (a) and parallel (b) compositions of Definition \ref{Def_GSP}.

In fact, there is a close relationship between $GSP$ and $SP$ graphs, which is reflected in the following lemma:

\begin{lemma}[adapted from Lemma 3.2 in \citet{Ho1999}] \label{Lemma_blocks_are_SP}
A connected graph $G$ is a $GSP$ graph if and only if each block of $G$ (i.e. each maximal induced biconnected subgraph of $G$) is an $SP$ graph.
\end{lemma}

\section{Results}

\subsection{Reducing phylogenetic networks to related graphs} \label{NetworkReduction}
In the following, we will introduce three ways to reduce phylogenetic networks to related simple graphs, which will play a crucial role in what follows.

\paragraph*{Leaf cutting}: Let $N^u$ be a phylogenetic network on taxon set $X$ with at least two vertices, at least two of which are leaves, i.e. $|V(N^u)|\geq 2$, $|X|\geq 2$. Let $G$ be the simple graph resulting from deleting all leaves labelled by $X$ from $V(N^u)$ and their incident edges. (Note that this may result in some vertices of degree 2 and -- e.g. if $N^u$ is a tree -- even new leaves not labelled by $X$, which we do \emph{not} remove). We call the simple graph resulting from this procedure the \emph{leaf cut graph of $N^u$} and denote it by $\mathcal{LCUT}(N^u)$. An illustration of the described procedure is depicted in Figure \ref{leafcutting}.

\begin{figure}[H]
\centering
\includegraphics[scale=0.4]{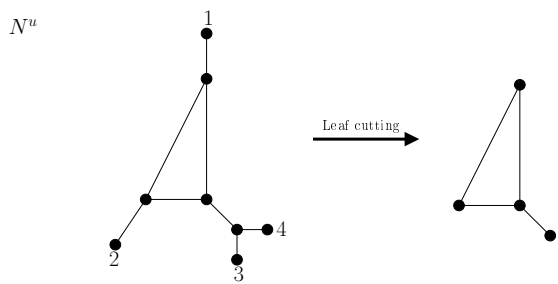}
\caption{Network $N^u$ on labelset $X=\{1,2,3,4\}$ and the simple graph resulting from the leaf cutting procedure. Note that this procedure results in one new leaf not labeled by $X$.}
\label{leafcutting}
\end{figure}

Based on the leaf cutting procedure, we can define a special class of phylogenetic networks, namely so-called \emph{$\mathcal{H}$-connected} networks, which will be of interest later on (e.g. in Section \ref{classes}).

\begin{definition} \label{def_Hconnected}
Let $N^u$ be a proper phylogenetic network on leaf set $X$ with $|X| \geq 2$ such that $\mathcal{LCUT}(N^u)$ is Hamilton connected. Then, $N^u$ is called an \emph{$\mathcal{H}$-connected} network.
\end{definition}

We can now turn to a second network reduction procedure, namely to the leaf shrinking procedure. We will apply this procedure not only to phylogenetic networks, but more generally to connected graphs, which is why we directly define it for general graphs.

\paragraph*{Leaf shrinking}: Let $G$ be a connected graph with at least two vertices, at least two of which are leaves, i.e. $|V(G)|\geq 2$, $|V_L(G)|\geq 2$. We shrink $G$ to a smaller simple and simple graph by constructing restricted topological subgraphs as described in Section \ref{graphtheo}, i.e. we delete vertices of degree 1, suppress vertices of degree 2 and delete a copy of parallel edges or loops. We do this as follows:

\begin{algorithm}[H] \label{alg_edgebased}
\SetKwInOut{Input}{input}\SetKwInOut{Output}{output}
\Input{Connected graph $G$ (e.g. a phylogenetic network $N^u$) with $\vert V(G) \vert \geq 2$ and $\vert V_L(G) \vert \geq 2$ }
\Output{Leaf shrink graph $\mathcal{LS}(G)$ of $G$}
$\mathcal{LS}(G) \coloneqq G$\;
	\While{$\vert V(\mathcal{LS}(G)) \vert > 2$}{
	Do one of the following (if applicable):
		\begin{itemize}
		\item Delete a leaf $x$ (and its incident edge) from $\mathcal{LS}(G)$.
		\item Suppress a vertex of degree 2 in $\mathcal{LS}(G)$.
		\item Delete one copy of a multiple edge, i.e. if $e' = e \in E(\mathcal{LS}(G))$, delete $e$.
		\item Delete a loop, i.e. if $e=\{u,u\} \in E(\mathcal{LS}(G))$, delete $e$.
		\end{itemize}
	If no operation is applicable, \Return $\mathcal{LS}(G)$
	}
	\If{$\vert V(\mathcal{LS}(G)) \vert = 2$}{
		\While{$\vert E(\mathcal{LS}(G)) \vert > 1$}{
			\If{$G$ contains a multiple edge (i.e. if $e' = e \in E(\mathcal{LS}(G))$) or if $\mathcal{LS}(G)$ contains a loop (i.e. if $e=\{u,u\} \in E(\mathcal{LS}(G))$)}{
		Delete $e$.
			}
		}
	}
	\Return $\mathcal{LS}(G)$
	\caption{Leaf shrinking}
\end{algorithm}

We call a simple graph resulting from this procedure the \emph{leaf shrink graph of} $G$ and denote it by $\mathcal{LS}(G)$. Note that due to steps $5-11$ in the algorithm, the smallest graph (in terms of the number of vertices and the number of edges) a graph $G$ may be reduced to is the complete graph on 2 vertices $K_2$, i.e. a single edge.

Based on the leaf shrinking procedure, we can again introduce a special class of phylogenetic networks, namely so-called edgebased phylogenetic networks (cf. Figure \ref{leafshrinking}). We will elaborate on edgebased phylogenetic networks in Section \ref{sec_edgebased}.

\begin{definition} \label{def_edgebased}
Let $G$ be a connected graph with $\vert V(G) \vert \geq 2$ and $\vert V_L(G) \vert \geq 2$. If the leaf shrink graph $\mathcal{LS}(G)$ of $G$ is a single edge, $G$ is called \emph{edgebased}. Else, $G$ is called \emph{non-edgebased}. 
If $G = N^u$ is a proper phylogenetic network with $|V(N^u)| \geq 2$ and $|X| \geq 2$ and $\mathcal{LS}(N^u)$ is a single edge, we call $N^u$ an \emph{edgebased} network. Else, $N^u$ is called \emph{non-edgebased}.
\end{definition}

\begin{remark} 
Note that the definition of edgebased graphs is very similar to the definition of $GSP$ graphs (cf. page \pageref{gsp}), the only difference being that a fourth operation -- the deletion of loops -- is allowed. In Section \ref{sec_edgebased} we will, however, show that there is a direct relationship between these two classes of graphs.  
\end{remark}

\begin{figure}[H]
\centering
\includegraphics[scale=0.4]{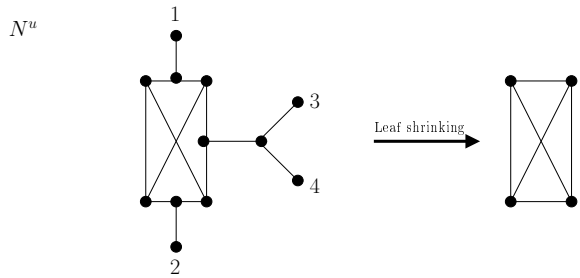}
\caption{Network $N^u$ on labelset $X=\{1,2,3,4\}$ and the simple graph resulting from the leaf shrinking procedure. At first, leaves 1, 2, 3 and 4 are deleted, resulting in a graph with one new leaf without label, which is subsequently removed. Afterwards, all resulting degree-2 vertices are suppressed.}
\label{leafshrinking2}
\end{figure}
\begin{figure}[H]
\centering
\includegraphics[scale=0.4]{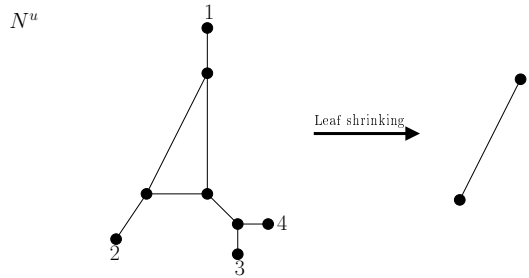}
\caption{Network $N^u$ on labelset $X=\{1,2,3,4\}$ and the simple graph resulting from the leaf shrinking procedure, which is an edge. At first, leaves 1, 2, 3 and 4 are deleted, resulting in a graph with one new leaf without label (cf. Figure \ref{leafcutting}). Then, this new leaf is removed as well, which results in a triangle. Now, one vertex of degree 2 is suppressed and the parallel edge is deleted resulting in one single edge. Thus, $N^u$ is called edgebased. Note that this graph resulting from the leaf shrinking procedure differs from the graph resulting from the leaf cutting procedure depicted in Figure \ref{leafcutting}.}
\label{leafshrinking}
\end{figure}

The last network reduction procedure that we want to introduce is the so-called \emph{leaf connecting} procedure.

\paragraph*{Leaf connecting}: Let $N^u$ be a phylogenetic network that is not a tree\footnote{Note that for a tree, the pre-processing step would always result in a single edge.} on taxon set $X$ with at least two leaves, i.e. $|X|\geq 2$. Then, we turn $N^u$ into a simple graph without vertices of degree 1 as follows: First of all, as a pre-processing step, as long as there exists an internal vertex $v$ of $N^u$ such that there is more than one leaf attached to $v$, delete all but one of the leaves adjacent to $v$.
If this results in $deg(v)=2$, suppress $v$. Note that this can only happen if $v$ is adjacent to only one inner vertex of $N^u$ and at least two leaves. In particular, this implies that suppressing $v$ cannot lead to parallel edges (cf. Figure \ref{Fig_LeafConnecting}, where in the pre-processing step vertex $x$ is suppressed).

Note that this pre-processing step may have to be repeated several times, but does not influence the property of a network being treebased or not. If a network is treebased, there exists a base tree that in particular covers all leaves attached to some vertex $v$. By deleting all but one of them and suppressing resulting degree-2 vertices, we obtain a base tree for the pre-processed network. Conversely, given a base tree for a pre-processed network, we can obtain a base tree for the original network by subdividing edges (if necessary) and adding leaves to these attachment points or to existing vertices of the base tree. 

After the pre-processing step, we continue as follows:
\begin{itemize}
\item Select two leaves $x_1$ and $x_2$ (if they exist). We call their respective attachment points $u_1$ and $u_2$, respectively. Delete $x_1$ and $x_2$ as well as edges $\{x_1,u_1\}$ and $\{x_2,u_2\}$ and add an edge $e:=\{u_1,u_2\}$. If this edge is a parallel edge, i.e. if there is another edge $\widetilde{e}$ connecting $u_1$ and $u_2$, add two more vertices $a$ and $b$ and replace $e$ by two new edges, namely $e_1:=\{u_1,a\}$ and $e_2:=\{a,u_2\}$. Similarly, replace $\widetilde{e}$ by two new edges, namely $\widetilde{e}_1:=\{u_1,b\}$ and $\widetilde{e}_2:=\{b,u_2\}$. Last, add a new edge $\{a,b\}$. \\
Repeat this procedure until no pair of leaves is left.
\item  If there is one more leaf $x$ left in the end, remove $x$ and, if its attachment point $u$ then has degree 2, suppress $u$. If this results in two parallel edges, say $e=\{y,z\}$ and $\widetilde{e} = \{y,z\}$, re-introduce $u$ on edge $e$ and add a new vertex $a$ to the graph, delete $\widetilde{e}$ and introduce two new edges $\widetilde{e}_1:=\{y,a\}$ and $\widetilde{e}_2:=\{a,z\}$. Last, add an edge $\{u,a\}$.
\end{itemize}
Note that the order in which the leaves are joined may alter the resulting graph. So if $|X|>2$, there might be more than one graph that can be achieved from $N^u$ in this manner. We refer to the set of such graphs as $\mathcal{LCON}(N^u)$.
Two illustrations of this concept are given in Figures \ref{Fig_LeafConnecting} and \ref{Fig_LCON}, respectively.

\begin{figure}[htbp]
	\centering
	\includegraphics[scale=0.35]{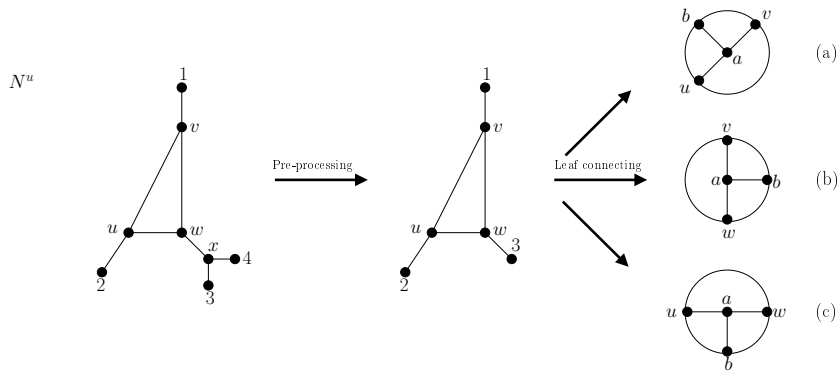}
	\caption{Network $N^u$ and the simple graphs resulting from the leaf connecting procedure. First, according to the pre-processing phase of the leaf connecting procedure, leaf 4 is deleted from the network because $x$ is adjacent to two leaves. Then, $x$ has degree 2 and thus needs to be suppressed. Then, first a pair of leaves is chosen and removed from the network, before the last leaf is removed (for (a), first leaves $1$ and $2$ were removed, followed by leaf $3$; for (b), first leaves $1$ and $3$ were removed, followed by leaf $2$ and for (c), first leaves $2$ and $3$ were removed, followed by leaf $1$). Note that the graphs depicted in (a), (b) and (c) are isomorphic. Thus, here $\mathcal{LCON}(N^u)$ consists of a single simple graph (in general, $\mathcal{LCON}(N^u)$ can consist of several simple graphs; as an example see Figure \ref{Fig_LCON}). Note, however, that the simple graph in $\mathcal{LCON}(N^u)$ differs from the simple graphs obtained from the leaf cutting and leaf shrinking procedure (cf. Figures \ref{leafcutting} and \ref{leafshrinking}).
	Moreover, note that even though new vertices ($a$ and $b$) were introduced, the total number of vertices of the simple graph in $\mathcal{LCON}(N^u)$ did not increase compared to $N^u$ or even compared to $N^u$ after the pre-processing step.}
	\label{Fig_LeafConnecting}
\end{figure}

\begin{figure}[htbp]
	\centering
	\includegraphics[scale=0.18]{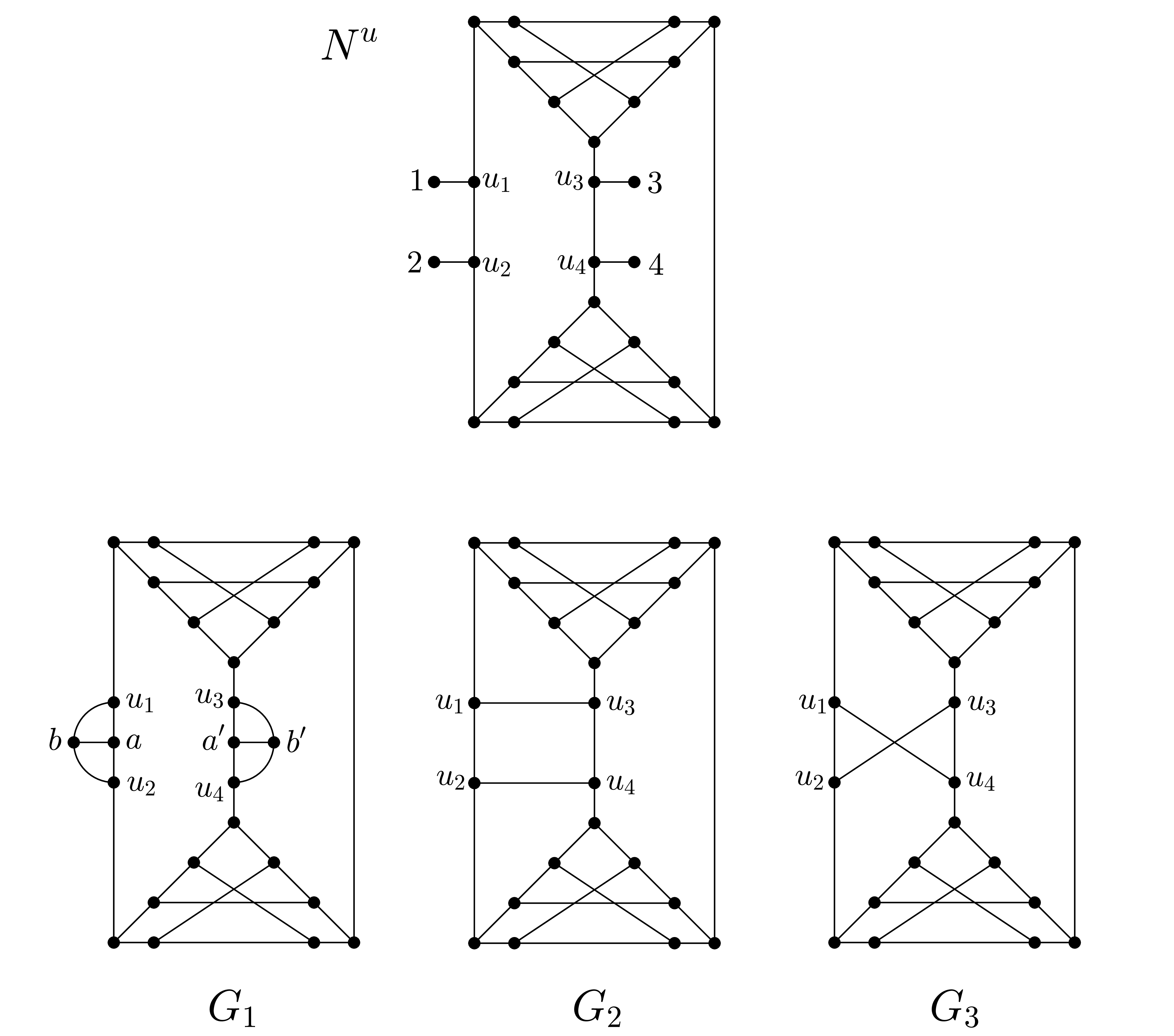}
	\caption{Network $N^u$ (adapted from \citet{paper2}) and the set $\mathcal{LCON}(N^u)$ resulting from the leaf connecting procedure. $G_1$ is obtained by deleting leaves $1,2$ and $3,4$ and connecting their attachment points respectively, while $G_2$ is obtained by connecting leaves $1,3$ and $2,4$ and $G_3$ is obtained by connecting leaves $1,4$ and $2,3$. Note that in case of $G_1$, 4 vertices ($a,b,a',b'$) have to be introduced in order to prevent the graph from becoming a multigraph. For $G_2$ and $G_3$ this step is not necessary. Note, however, that in any case the number of vertices of a graph in $\mathcal{LCON}(N^u)$ cannot increase compared to $N^u$, because in each step 2 leaves are deleted and at most 2 new vertices are created.}
	\label{Fig_LCON}
\end{figure}

To summarize, leaf cutting, leaf shrinking and leaf connecting are three different procedures to reduce a phylogenetic network to related simple graphs. In general, the resulting graphs differ from each other. However, all of them lead to sufficient criteria for treebasedness, which we will introduce in the following. We start with considering of edgebased phylogenetic networks in more detail.

\subsection{Classes of treebased networks} \label{sec_classes}
Deciding whether an unrooted phylogenetic network is treebased or not is generally NP-complete (cf. \citet{Francis2018}). So for practical purposes, it would be useful to know some sufficient properties that can be checked in polynomial time and which assure that a given network is indeed treebased (even if these criteria are not necessary). In this section, we will introduce a class of unrooted phylogenetic networks which are guaranteed to be treebased, namely edgebased networks. While edgebasedness can be checked in linear time, we will additionally mention other classes of networks which are also guaranteed to be treebased, but are based on properties like being Hamiltonian or Hamilton connected. However, while it is generally hard to check these properties (cf. \citet{Karp1972}), graphs with these properties have been extensively studied, which is why such properties link phylogenetic network theory to classical graph theory. Moreover, various graphs are already known to be Hamiltonian or Hamilton connected (cf. eg. \citet{Wilson1988,Rahman2005,Zhao2007} and \citet{Hu2005,Alspach2013}). Therefore, these properties might help to further enhance the understanding of phylogenetic networks. 

\subsubsection{Edgebased networks} \label{sec_edgebased}
In this section we thoroughly analyze the class of edgebased graphs and networks, where our overall aim is to show that edgebasedness guarantees treebasedness. However, we first show that there is a direct relationship between loopless edgebased graphs and $GSP$ graphs. We then show that the order of restriction operations does not matter for neither of them, in the following sense: If a graph $G$ is edgebased (or a $GSP$ graph), not only does there exist a sequence of restriction operations that reduces $G$ to $K_2$, but \emph{any} sequence of restriction operations will lead to a graph on two vertices that can then be further reduced to $K_2$ (cf. Algorithm \ref{alg_edgebased}). Last but not least, we return to the phylogenetic setting and finally show that edgebased networks are always treebased.

\paragraph{Relationship between edgebased graphs and $GSP$ graphs}
Comparing the definitions of $GSP$ graphs (cf. page \pageref{gsp}) and edgebased graphs (cf. page \pageref{def_edgebased}) there is a slight difference between the two classes of graphs. Both can be reduced to a single edge by certain restriction operations, but while in the case of edgebased graphs the deletion of loops is a valid restriction operation, it is not in the case of $GSP$ graphs. In the following we will show, however, that there is a direct relationship between both classes of graphs.

\begin{theorem} \label{thm_gsp_edgebased}
Let $G$ be a connected graph. Then $G$ is a $GSP$ graph if and only if 
	\begin{enumerate}
	\item [\rm (i)] $G$ is loopless and
	\item [\rm (ii)] $G$ can be reduced to $K_2$ by deleting leaves, suppressing vertices of degree 2, deleting copies of parallel edges and deleting loops, i.e. by applying restriction operations 1. to 4. (cf. page \pageref{restriction_operations}). 
	\end{enumerate}
\end{theorem}

\begin{proof}
First, suppose that $G$ is a $GSP$ graph. Then $G$ does by definition not contain loops, i.e. {\rm (i)} holds. Moreover, by definition $G$ can be reduced to $K_2$ by applying restriction operations 1. to 3. (cf. page \pageref{restriction_operations}), so {\rm (ii)} holds as well.

Now, suppose that $G$ is a connected graph without loops that can be reduced to $K_2$ by applying restriction operations 1. to 4. In order to show that $G$ is a $GSP$ graph, we need to show that $G$ can also be reduced to $K_2$ by \emph{only} applying operations 1. to 3., i.e. by deleting leaves, suppressing degree-2 vertices and deleting copies of parallel edges, but \emph{not} deleting loops.  Now, as $G$ is by assumption a graph without loops, loops can only arise during the reduction process. Let $\widetilde{G}$ be a restricted topological subgraph of $G$ that contains a loop and assume that $\widetilde{G}$ is the first graph with loops that arises when reducing $G$ to $K_2$. This implies that on the way from $G$ to $\widetilde{G}$ there must have been a restricted topological subgraph $G'$ of $G$ containing a parallel edge $e=\{u,v\}$ where one of $u$ and $v$, say $v$, was a degree-2 vertex, and the step from $G'$ to $\widetilde{G}$ was the suppression of $v$.
Now, deleting the loop $\{u,u\}$ from $\widetilde{G}$ yields some restricted topological subgraph $\widehat{G}$ of $G$. 
However, $\widehat{G}$ can alternatively be reached from $G'$ by first deleting a copy of the parallel edge $e=\{u,v\}$ (yielding a graph $G''$) and then deleting vertex $v$. Thus, $\widehat{G}$ can be reached from $G$ by only applying operations 1. to 3. (cf. Figure \ref{Fig_LoopDelete}). As the deletion of loops can always be circumvented in this way, $G$ can in particular be reduced to $K_2$ by only applying operations 1. to 3. Together with the fact that $G$ is loopless, this implies that $G$ is a $GSP$ graph. This completes the proof.
\end{proof}

\begin{figure}[htbp]
\centering
\includegraphics[scale=0.2]{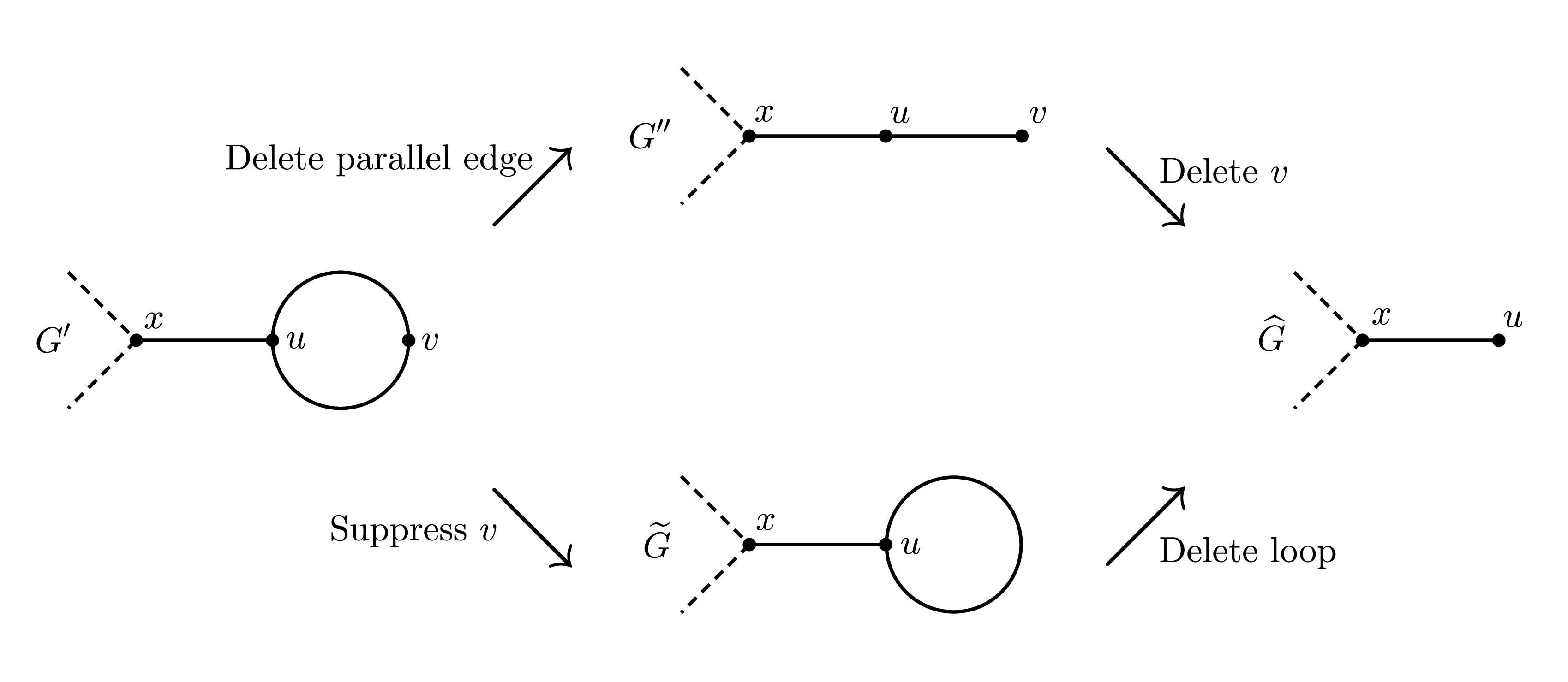}
\caption{Two alternative ways to reach graph $\widehat{G}$ from $G'$ (and thus $G$).}
\label{Fig_LoopDelete}
\end{figure}

As the following corollary shows, Theorem \ref{thm_gsp_edgebased} implies that there is a one-to-one relationship between loopless edgebased graphs and $GSP$ graphs.

\begin{corollary} \label{cor_gsp_edgebased}
Let $G$ be a connected graph. Then $G$ is a $GSP$ graph if and only if it is loopless and edgebased.
\end{corollary}

\begin{proof}
First, suppose that $G$ is a $GSP$ graph. Then due to Theorem \ref{thm_gsp_edgebased}, $G$ is loopless and can be reduced to $K_2$ by deleting leaves, suppressing degree-2 vertices, deleting copies of parallel edges and deleting loops. Let $\widehat{G}$ be a restricted topological subgraph of $G$ with $\vert V(\widehat{G}) \vert = 2$. Then, either $\widehat{G} = K_2$ or $\widehat{G}$ can be reduced to $K_2$. However, the latter reduction cannot require the deletion of leaves or suppression of degree-2 vertices (as this would reduce the number of vertices to less than 2, in which case $K_2$ could not be a restricted topological subgraph). This implies that $G$ can be reduced to $K_2$ by applying Algorithm \ref{alg_edgebased}, and thus $G$ is edgebased.

Now, suppose that $G$ is loopless and edgebased. The latter implies that $G$ can be reduced to $K_2$ by applying Algorithm \ref{alg_edgebased}. Together with Theorem \ref{thm_gsp_edgebased} and the fact that $G$ is loopless, this implies that $G$ is a $GSP$ graph, which completes the proof.
\end{proof}

Note that $GSP$ graphs can be recognized in linear time (cf. \citet{Wimer1988, Ho1999}). A naive approach would for example be to consider the maximal biconnected components (or blocks) of a graph $G$ (which can be computed in linear time (cf. \citet{Hopcroft1973}) and use the fact that a graph $G$ is a $GSP$ graph if and only if each block of $G$ is an $SP$ graph (cf. Lemma \ref{Lemma_blocks_are_SP}), where the latter can again be recognized in linear time (cf. \citet{Valdes1982}). Due to the one-to-one relationship between $GSP$ graphs and loopless edgebased graphs this implies that edgebasedness can also be tested in linear time. In particular, it can be decided in linear time whether an unrooted phylogenetic network is edgebased or not. As we will later on show that edgebasedness implies treebasedness (cf. Theorem \ref{thm-edgebased}), this is of great relevance, since, in general, the problem of deciding whether a network is treebased or not is an NP-complete problem (cf. \citet{Francis2018}).

However, before analyzing the relationship between edgebasedness and treebasedness, we first state another interesting property of edgebased graphs and $GSP$ graphs, namely that the order of restriction operations does not matter.

\paragraph{Order of restriction operations}

\begin{theorem} \label{thm-edgebased}
Let $G$ be an edgebased graph. Then \emph{all} sequences of restriction operations in concordance with Algorithm \ref{alg_edgebased} lead to $K_2$.
\end{theorem}

\begin{remark}
Theorem \ref{thm-edgebased} implies that the order of restriction operations does not matter as long as it follows the rules of Algorithm \ref{alg_edgebased}, i.e. whenever there are two or more operations possible, it does not matter which one is chosen. Recall, however, that as soon as $\vert V(\mathcal{LS}(G)) \vert = 2$, the choice of restriction operations is limited to deleting copies of parallel edges or deleting loops in order to prevent the number of vertices dropping to less than two.
\end{remark}

The proof of Theorem \ref{thm-edgebased} requires the following lemmas. 

\begin{lemma}\label{addsth} 
Let $G$ be a graph with vertex set $V(G)$ and edge set $E(G)$ such that $K_2$ is a restricted topological subgraph of $G$. Let $G'$ result from $G$ by precisely one of the following operations:

\begin{enumerate}
\item Choose a vertex $u \in V(G)$, introduce a new vertex $x$ and an edge $\{u,x\}$ (`Add leaf $x$').
\item Choose an edge $e \in E(G)$ and subdivide it into two edges by introducing a new degree-2 vertex (`Add a degree-2 vertex').
\item Choose an edge $e \in E(G)$ and add a copy $e'$ of $e$ to $E(G)$.
\item Choose a vertex $u \in V(G)$ and add a loop, i.e. add edge $e=\{u,u\}$ to $E(G)$.
\end{enumerate}
Then, $K_2$ is also a restricted topological subgraph of $G'$.
\end{lemma}

\begin{proof} We can convert $G'$ into $G$ by undoing the respective operation. Then, because $G$ can be reduced to $K_2$, so can $G'$ (using the conversion to $G$ as a first step and adding the sequence that converts $G$ to $K_2$). 
This completes the proof.
\end{proof}

The proofs of the following two lemmas can be found in the appendix.

\begin{lemma} \label{loopdelete} 
Let $G$ be a connected graph with vertex set $V(G)$ and edge set $E(G)$. Let $G'$ result from $G$ by deleting one loop. Then, $K_2$ is a restricted topological subgraph of $G$ if and only if $K_2$ is a restricted topological subgraph of $G'$.
\end{lemma}

\begin{lemma} \label{paralleldelete} 
Let $G$ be a connected graph with vertex set $V(G)$ and edge set $E(G)$. Let $G'$ result from $G$ by deleting one copy of a parallel edge. Then, $K_2$ is a restricted topological subgraph of $G$ if and only if $K_2$ is a restricted topological subgraph of $G'$.
\end{lemma}

The last two lemmas immediately lead to the following corollary, which plays a fundamental role in the proof of Theorem \ref{thm-edgebased}.

\begin{corollary}\label{wlogsimple}
Let $G$ be a graph and let $G'$ be its underlying simple graph. Then, $K_2$ is a restricted topological subgraph of $G$ if and only if $K_2$ is a restricted topological subgraph of $G'$.
\end{corollary}

\begin{proof} $G'$ is like $G$ but without parallel edges and without loops. Now if $G'$ has $K_2$ as a restricted topological subgraph, by repeatedly applying Operations 3 and 4 of Lemma \ref{addsth}, so does $G$. On the other hand, if $G$ has $K_2$ as a restricted topological subgraph, by repeatedly applying Lemmas \ref{loopdelete} and \ref{paralleldelete}, so does $G'$. This completes the proof.
\end{proof}

We are finally in a position to prove Theorem \ref{thm-edgebased}. 

\begin{proof}[Theorem \ref{thm-edgebased}]
Let $G$ be an edgebased graph and assume that there exists a sequence $\sigma$ of restriction operations in concordance with Algorithm \ref{alg_edgebased} that does not lead to $K_2$. This implies that $G$ has $K_2$ as a restricted topological subgraph (since it is edgebased), but it also has some restricted topological subgraph that does \emph{not} have $K_2$ as a restricted topological subgraph (since $\sigma$ does not lead to $K_2$).

We take a minimal graph with this property in terms of the number of vertices. So we assume that $G$ has $K_2$ as a restricted topological subgraph, but that there exists a restricted topological subgraph $G'$ of $G$ which does not have $K_2$ as a restricted topological subgraph, and that there is no other graph with this property with fewer vertices than $G$. Due to Corollary \ref{wlogsimple}, we may assume that $G$ has no loops and no parallel edges.

We now consider the reduction of $G$ to $G'$. As $G$ has no parallel edges and no loops, the first step on the way from $G$ to $G'$ must be the deletion of a leaf or the suppression of a degree-2 vertex. Moreover, the resulting graph $G''$ after one step must already be such that $K_2$ is not a restricted topological subgraph: Otherwise, $G''$ would also have $G'$ as a restricted topological subgraph (as it is on the path from $G$ to $G'$), it would have $K_2$ as a restricted topological subgraph and it would have strictly fewer vertices than $G$, which would contradict the minimality of $G$. 

Let us now consider $G''$. So $G''$ can be reached from $G$ by deleting a leaf $x$ or suppressing a vertex $u$ of degree 2, and while $K_2$ is a restricted topological subgraph of $G$, it is not of $G''$. Moreover, we consider $\widetilde{G}$, which shall be a graph that can be reached from $G$ within one step (i.e. after performing one restriction operation) on the way from $G$ to $K_2$. As $\widetilde{G}$ has $K_2$ as a restricted topological subgraph, and as $\widetilde{G}$ has strictly fewer vertices than $G$, we know that \emph{all} restricted topological subgraphs of $\widetilde{G}$ have $K_2$ as a restricted topological subgraph. 

We now consider the case that a leaf $x$ has been deleted on the way from $G$ to $G''$. Note that $x$ is also present in $\widetilde{G}$, as by any other restriction operation than the deletion of $x$, $x$ cannot be affected (as $G''$ and $\widetilde{G}$ cannot be equal and both differ from $G$ by the removal of precisely one vertex). So now we delete $x$ from $\widetilde{G}$ to obtain a graph $\widehat{G}$, which has $K_2$ as a restricted topological subgraph. By Lemma \ref{addsth}, we can undo the step that has been done on the way from $G$ to $\widetilde{G}$, i.e. we can re-add the leaf that has been deleted or the suppressed degree-2 vertex to $\widehat{G}$, and the resulting graph -- which, by the way, is precisely $G''$ -- has $K_2$ as a restricted topological subgraph. This is a contradiction to the construction of $G''$. 

If, on the other hand, a degree-2 vertex $u$ has been suppressed on the way from $G$ to $G''$, then either $u$ is still present as a degree-2 vertex in $\widetilde{G}$, or $u$ is a leaf in $\widetilde{G}$ (if a leaf adjacent to $u$ has been deleted). In the first case, i.e. if $u$ still has degree 2 in $\widetilde{G}$, we can suppress $u$ in order to obtain a graph $\widehat{G}$, which has $K_2$ as a restricted topological subgraph. By Lemma \ref{addsth}, we can undo the step that has been done on the way from $G$ to $\widetilde{G}$, i.e. we can re-add the leaf that has been deleted or the suppressed degree-2 vertex to $\widehat{G}$, and the resulting graph -- which, by the way, is precisely $G''$ -- has $K_2$ as a restricted topological subgraph. This is a contradiction to the construction of $G''$. 

So the only remaining case is the case that a degree-2 vertex $u$ has been suppressed on the way from $G$ to $G''$, and $u$ is a leaf in $\widetilde{G}$. However, this can only be the case if a leaf $x$ adjacent to $u$ has been deleted on the way from $G$ to $\widetilde{G}$. But if $u$ is a degree-2 vertex adjacent to a leaf, then deleting the leaf and its incident edge is equivalent to suppressing $u$, i.e. the resulting graphs $G''$ and $\widetilde{G}$ are isomorphic. This is illustrated by Figure \ref{Fig_SuppressDelete}. Thus, as $K_2$ is a restricted topological subgraph of $\widetilde{G}$, it is also of $G''$, but this contradicts the construction of $G''$.  

\begin{figure}[htbp]
\centering
\includegraphics[scale=0.2]{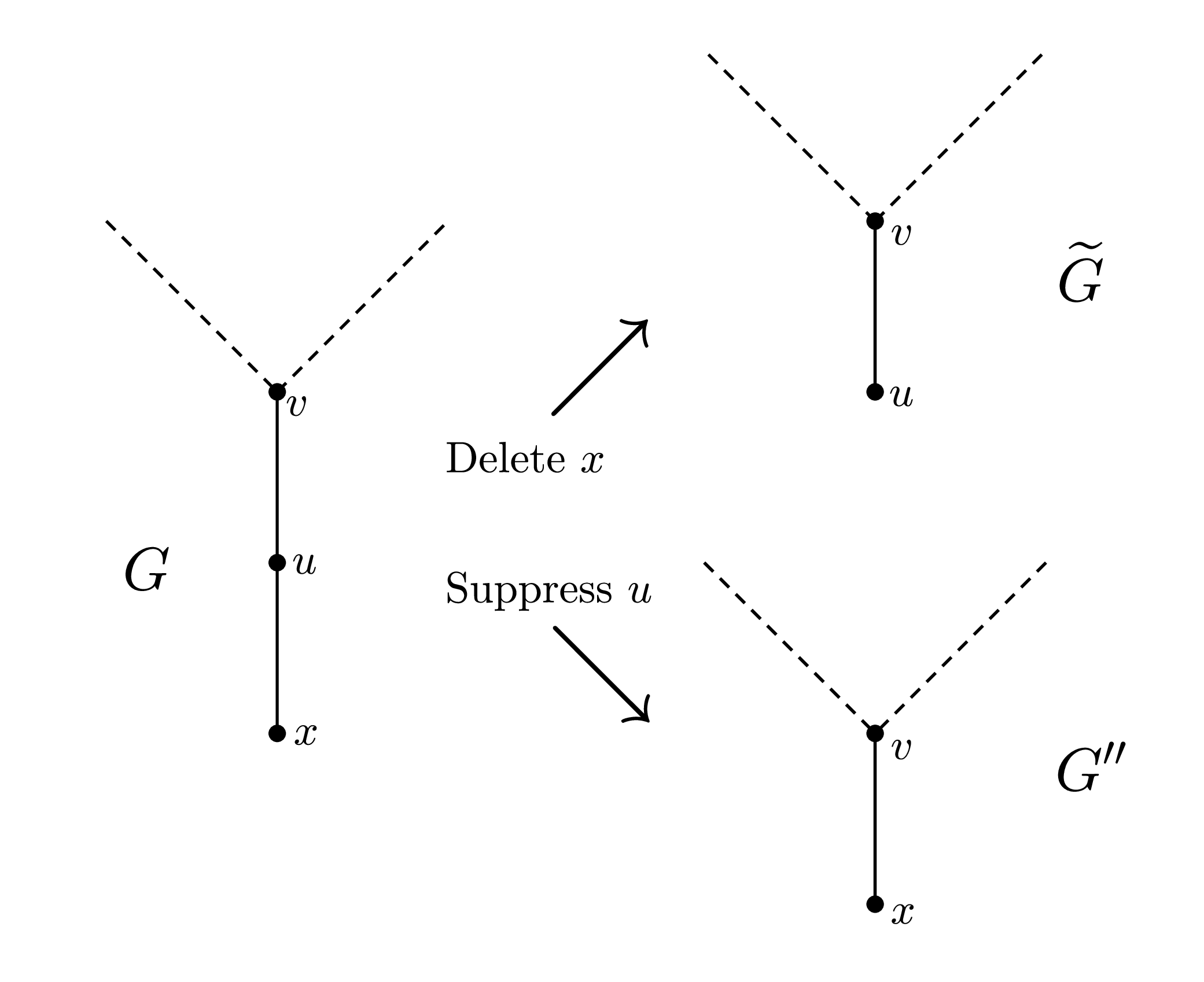} 
\caption{Two isomorphic graphs $\widetilde{G}$ and $G''$ that are constructed from $G$ by either deleting leaf $x$ or suppressing vertex $u$. }
\label{Fig_SuppressDelete}
\end{figure}

Therefore, all cases lead to a contradiction, which shows that the assumption was wrong. In particular, all sequences of restriction operations in concordance with Algorithm \ref{alg_edgebased} eventually lead to $K_2$. This completes the proof.
\end{proof}

\paragraph{Edgebasedness implies treebasedness}
We now state the last main theorem of this section, which shows that all edgebased networks (cf. Definition \ref{def_edgebased}) are also treebased.

\begin{theorem} \label{edgebased_implies_treebased}
Let $N^u$ be a proper phylogenetic network on leaf set $X$ with $\vert X \vert \geq 2$. If $N^u$ is edgebased, it is also treebased.
\end{theorem}

Note that the converse does not hold: Figure \ref{leafshrinking2} for example shows a treebased network $N^u$ that is not edgebased. \\

In order to prove Theorem \ref{edgebased_implies_treebased}, we will exploit the one-to-one relationship between loopless edgebased graphs and $GSP$ graphs (cf. Corollary \ref{cor_gsp_edgebased}). Moreover, we will use the fact that a graph is a $GSP$ graph if and only if its blocks are $SP$ graphs (cf. Lemma \ref{Lemma_blocks_are_SP}).

The strategy for the proof of Theorem \ref{edgebased_implies_treebased} is thus to decompose an edgebased network $N^u$ into its blocks (which are $SP$ graphs due to Lemma \ref{Lemma_blocks_are_SP}, since $N^u$ is loopless by definition and thus a GSP graph by Corollary \ref{cor_gsp_edgebased}), infer a certain spanning tree for each block and use these spanning trees to build a support tree for $N^u$. This requires the following additional technical lemma, the proof of which is given in the Appendix.

\begin{lemma} \label{Lemma_SP_SpanningTrees}
Let $G=(V,E)$ be a simple and biconnected $SP$ graph with at least three vertices. Then, there exists a spanning tree $T$ in $G$ whose leaves correspond to degree-2 vertices of $G$. In particular, no vertex $v \in V$ of $G$ with $deg(v) > 2$ is a leaf in $T$. 
\end{lemma}

\begin{remark} \label{remark_valid}
In the following, given a simple and biconnected $SP$ graph $G$ with at least three vertices, we call a spanning tree $T$ having only degree-2 vertices of $G$ as leaves a \emph{valid} spanning tree. Additionally, given the trivial $SP$ graph $K_2$, we also call a spanning tree for $K_2$ (which is $K_2$ itself) a valid spanning tree.
\end{remark}

With this we are now in a position to prove Theorem \ref{edgebased_implies_treebased}.

\begin{proof}[Proof of Theorem \ref{edgebased_implies_treebased}]
Let $N^u$ be a proper phylogenetic network on leaf set $X$ with $\vert X \vert \geq 2$. If $\vert V(N^u) \vert = \vert X \vert = 2$ and $N^u$ consists of a single edge, $N^u$ is trivially treebased. Thus, assume that $\vert V(N^u) \vert \geq 3$.

As $N^u$ is edgebased and loopless it is a $GSP$ graph by Corollary \ref{cor_gsp_edgebased} and we can decompose it into its blocks, i.e. into its maximal biconnected subgraphs (cf. Figure \ref{Fig_Blocks}). Now, due to Lemma \ref{Lemma_blocks_are_SP} these blocks are all $SP$ graphs. To be precise, each block of $N^u$ is either a trivial $SP$ graph, i.e. a single edge corresponding to a cut edge of $N^u$, or it is a simple and biconnected $SP$ graph with at least three vertices.

We now consider all blocks $\mathcal{B}$ of $N^u$ and construct a support tree $T$ for $N^u$ as follows:

If $\mathcal{B}=\{u,v\}$ is a single edge (i.e. $\mathcal{B}$ is a cut edge of $N^u$), we add this edge to $T$. Else, if $\mathcal{B}$ is a simple and biconnected $SP$ graph with at least 3 vertices we add all edges of a valid spanning tree $T_{\mathcal{B}}$ of $\mathcal{B}$, i.e. of a spanning tree of $\mathcal{B}$ having only degree-2 vertices of $\mathcal{B}$, to $T$ (which must exist due to Lemma \ref{Lemma_SP_SpanningTrees}).

Then, $T$ is a support tree for $N^u$, since:
	\begin{itemize}
	\item $T$ covers all vertices of $N^u$ (since it covers all vertices of each block $\mathcal{B}$ of $N^u$, respectively).
	\item $T$ is a tree, i.e. $T$ is connected and acyclic. To see this, note that any two blocks $\mathcal{B}_1$ and $\mathcal{B}_2$ of $N^u$ share at most one common vertex $v$ and if they share such a vertex $v$, $v$ is a cut vertex of $N^u$. Let $T_{\mathcal{B}_1}$ be a valid spanning tree of $\mathcal{B}_1$ and let $T_{\mathcal{B}_2}$ be a valid spanning tree for $\mathcal{B}_2$ (where both $T_{\mathcal{B}_1}$ and $T_{\mathcal{B}_2}$ are potentially single edges) and suppose that $\mathcal{B}_1$ and $\mathcal{B}_2$ share a common vertex $v$. Then identifying the copy of $v$ in $T_{\mathcal{B}_1}$ with the copy of $v$ in $T_{\mathcal{B}_2}$ yields a spanning tree for $\mathcal{B}_1 \cup \mathcal{B}_2$ (since identifying the two copies of $v$ cannot induce cycles since $\mathcal{B}_1$ and $\mathcal{B}_2$ and thus $T_{\mathcal{B}_1}$ and $T_{\mathcal{B}_2}$ do not share any other vertices than $v$).
	As every block of $N^u$ contains at least one cut vertex of $N^u$ and as $T$ covers all cut vertices of $N^u$, it iteratively follows that $T$ is connected and acyclic.
	\item The leaf set of $T$ corresponds to $X$. To see this, consider the leaves of the induced spanning trees $T_{\mathcal{B}}$ for each block $\mathcal{B}$ of $N^u$. 
		\begin{itemize}
		\item If $\mathcal{B}$ is a non-trivial $SP$ graph, its valid spanning tree $T_{\mathcal{B}}$ has only degree-2 vertices of $\mathcal{B}$ as leaves. Let $v$ be such a leaf. As $N^u$ does not contain degree-2 vertices (since it is a phylogenetic network), $v$ must be a cut vertex of $N^u$. However, with the argument from above $v$ is then contained in at least one other spanning tree $T_{\mathcal{B}'}$ for some other block $\mathcal{B}'$ of $N^u$ and thus cannot be a leaf in $T$ (since in $T$ the two copies of $v$ contained in $T_{\mathcal{B}}$ and $T_{\mathcal{B}'}$, respectively, are identified and thus $deg(v) \geq 2$ in $T$).
		\item Similarly, if $\mathcal{B}$ is a trivial $SP$ graph $\{u,v\}$ and if $\{u,v\}$ is an internal cut edge of $N^u$, neither $u$ nor $v$ can be leaves in $T$ (since again, both $u$ and $v$ are contained in at least one other spanning tree and after identifying all copies of $u$ and all copies of $v$, respectively, $deg(u), deg(v) \geq 2$ in $T$).
		\item Lastly, if $\mathcal{B} = \{x,v\}$ is a trivial $SP$ graph corresponding to an external cut edge of $N^u$, where $x \in X$ and $v$ is an internal vertex of $N^u$, $x$ is a leaf in $T$ and $v$ is an internal vertex in $T$. This is due to the fact that each leaf $x$ of $N^u$ is contained in exactly one block of $N^u$ (and thus will be a leaf in $T$ as there is only one copy of $x$), while there exists at least one other block $\mathcal{B}'$ containing a copy of $v$ and the two copies of $v$ will be identified in $T$.		
		\end{itemize}
	\end{itemize}
To summarize, $T$ is a spanning tree of $N^u$ which contains all leaves $x \in X$, but does not induce any additional leaves. Thus, $T$ is a support tree for $N^u$ and $N^u$ is treebased. This completes the proof.
\end{proof}

\begin{figure}[htbp]
\centering
\includegraphics[scale=0.3]{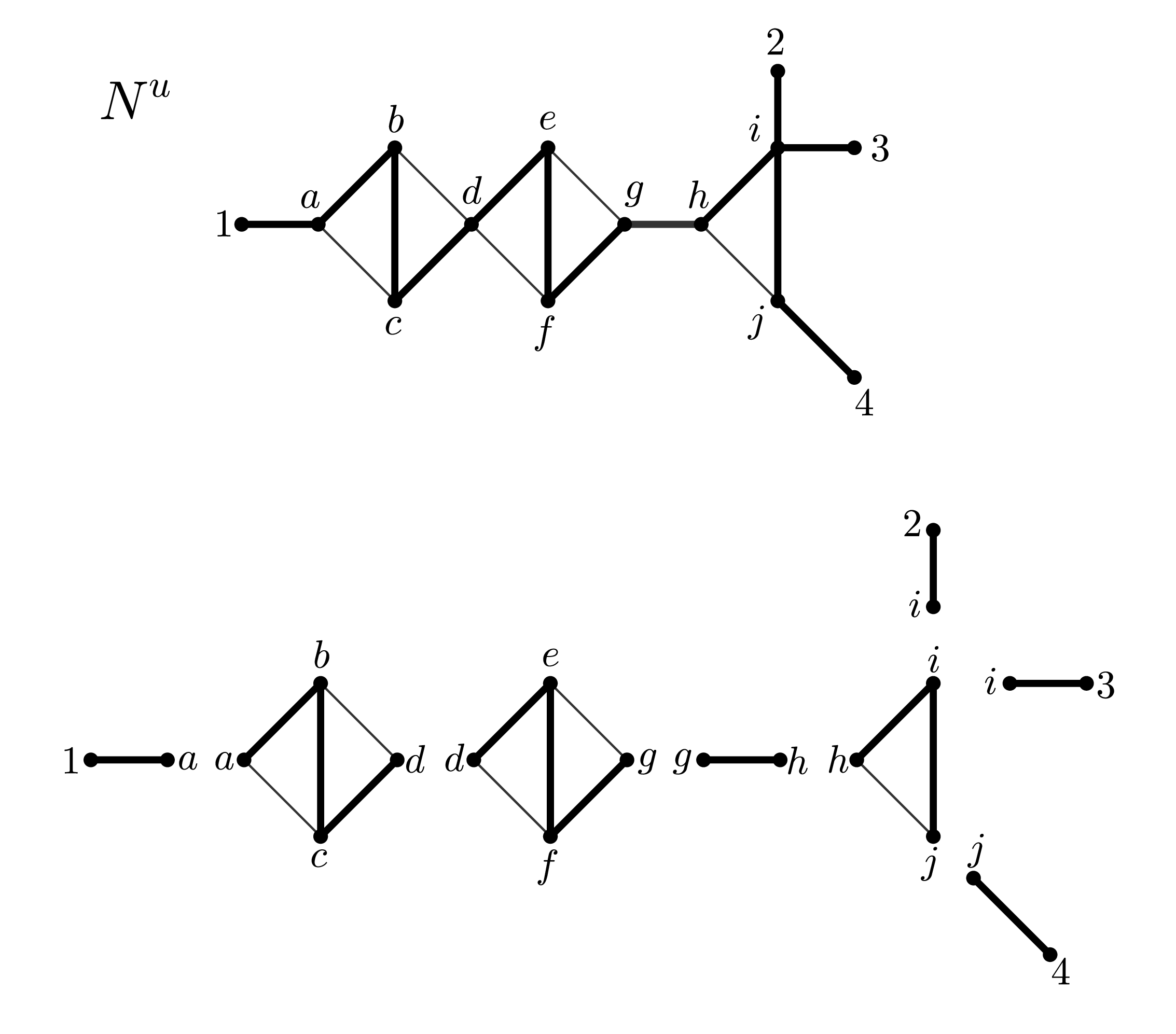}
\caption{Decomposition of the edgebased network $N^u$ into its maximal biconnected components that are either trivial $SP$ graphs, i.e. single edges (corresponding to cut edges of $N^u$), or simple and biconnected $SP$ graphs with at least three vertices. For the latter, valid spanning trees are depicted in bold, respectively. The edges of these valid spanning trees together with all cut edges (also depicted in bold) yield a support tree for $N^u$ and thus $N^u$ is treebased.}
\label{Fig_Blocks}
\end{figure}

To conclude, edgebased networks are always treebased and, more importantly, it can be checked in linear time whether a network is edgebased or not.

Additionally note that in order to check the edgebasedness of a network, we can make use of the fact that a network can be seen as a \enquote{blobbed} tree (cf. \citet{Gusfield2005}), i.e. as a tree with blobs as vertices (cf. p. \pageref{tree_of_blobs}). In particular, we have the following decomposition statement, which is the final statement of this section.

\begin{proposition} \label{blobreduce} Let $N^u$ be a proper unrooted phylogenetic network with at least two leaves. Then, $N^u$ is edgebased if and only if every non-trivial blob of $N^u$ is edgebased.
\end{proposition}

The proof of this proposition again exploits the one-to-one relationship between loopless edgebased graphs and $GSP$ graphs and uses the following theorem which implies that a $GSP$ graph can be reduced to any of its edges\footnote{See proof of Theorem 4.1 in \citet{Ho1999}.}.
\begin{theorem}[Theorem 4.1 in \citet{Ho1999}] \label{GSP_reduce_to_any_edge}
Let $G$ be a $GSP$ graph. Then, for any edge $e=\{u,v\}$ of $G$, $G$ is a $GSP$ graph with terminals $u$ and $v$.
\end{theorem}

We now use this theorem to prove Proposition \ref{blobreduce}.

\begin{proof}[Proof of Proposition \ref{blobreduce}]
First of all note that if $N^u$ contains only trivial blobs, it is a tree and is therefore trivially edgebased. Thus, we now consider the case that $N^u$ contains at least one non-trivial blob.
If $N^u$ is edgebased, then all non-trivial blobs of $N^u$ are also necessarily edgebased. If there was a non-trivial blob of $N^u$ with a restricted topological subgraph that cannot be reduced to an edge, this subgraph would also be contained as a restricted topological subgraph in $N^u$, which means that $N^u$ would have a restricted topological subgraph that cannot be reduced to a single edge. However, due to Theorem \ref{thm-edgebased} \emph{all} restricted topological subgraphs of $N^u$ must have a single edge as a restricted topological subgraph, so this is a contradiction. 

If, on the other hand, all non-trivial blobs of $N^u$ are edgebased, then we can show inductively that $N^u$ is edgebased. If $N^u$ contains only one non-trivial blob, there is nothing to show. Now assume that we know the statement is true for all networks with at most $m$ non-trivial blobs, and let $N^u$ contain $m+1$ non-trivial blobs. We now make use of the fact that $N^u$ must contain a cut edge $e=\{a,b\}$ whose removal results in two connected components, each containing at least one non-trivial blob. We denote these components by $N^u_a$ and $N^u_b$ and assume that $a$ is contained in $N^u_a$ and $b$ is contained in $N^u_b$. We now re-introduce the cut edge $\{a,b\}$ to both components by attaching a new leaf $a$ to $N_b^u$ and $b$ to $N_a^u$, respectively. Now, without loss of generality, we first consider $N^u_a$. As $N^u_a$ contains at most $m$ non-trivial blobs it is edgebased by the inductive hypothesis. Moreover, by Theorem \ref{GSP_reduce_to_any_edge}, we can reduce it to any of its edges, in particular to its leaf edge $e=\{a,b\}$. Similarly, as $N^u_b$ contains at most $m$ blobs it is also edgebased and can be reduced to its leaf edge $e=\{a,b\}$. In total, this implies that $N^u$ can be reduced to edge $e=\{a,b\}$. In particular, $N^u$ is edgebased. This completes the proof.
\end{proof}

\subsubsection{Other networks that are guaranteed to be treebased} \label{classes}
After having thoroughly analyzed edgebased networks, we will now turn to other classes of networks that are guaranteed to be treebased by using some classic graph theoretical arguments.

\begin{theorem} \label{thm-suff}Let $N^u$ be a proper phylogenetic network on leaf set $X$ with $|X|\geq 2$, and consider the graph $\mathcal{LCUT}(N^u)$ as well as the set $\mathcal{LCON}(N^u)$ as defined in Section \ref{NetworkReduction}. Then, the following statements hold:
\begin{enumerate}
\item If $N^u$ contains two leaves $x$ and $y$ with attachment points $u$ and $v$, respectively, such that the edge $\{u,v\}$ is contained in the edge set of $N^u$ and such that there is a path in $N^u$ from $u$ to $v$ visiting all inner vertices of $N^u$, then $N^u$ is treebased.
\item If $N^u$ is an $\mathcal{H}$-connected network (i.e. if $\mathcal{LCUT}(N^u)$ is Hamilton connected), then $N^u$ is treebased.
\item If there is a graph $G$ in $\mathcal{LCON}(N^u)$ such that $G$ is Hamiltonian and contains a Hamiltonian cycle which uses an edge of $G$ which is not contained in $N^u$ and which did not result from deleting the last leaf in case $|X^r|$ is odd (where $X^r$ denotes the reduced leaf set of $N^u$ after a potential pre-processing step), then $N^u$ is treebased.
\item  If there is a graph $G$ in $\mathcal{LCON}(N^u)$ such that $G$ is Hamiltonian and such that at least two new vertices, say $a$ and $b$, had to be added when connecting the attachment points $u$ and $v$ of two leaves $x$ and $y$ during the construction of $G$ in order to prevent parallel edges, then $N^u$ is treebased.
\end{enumerate}
\end{theorem}

\begin{figure}[htbp]
	\centering
	\includegraphics[scale=0.45]{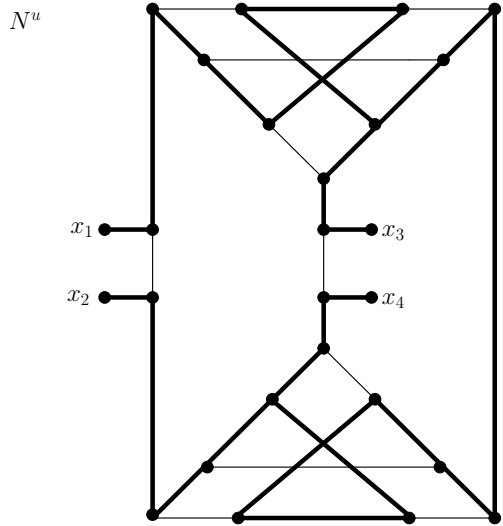}
	\caption{Binary treebased unrooted phylogenetic network $N^u$ on $X=\{x_1,x_2,x_3,x_4\}$. The corresponding support tree is highlighted in bold. $N^u-x_i$ is not treebased for $i=1,\ldots,4$, because there is no spanning tree in $N^u-x_i$ whose leaf set is equal to $X\setminus \{x_i\}$ (Figure taken from \citet{paper2}). }
	\label{Fig_notTreebased}
\end{figure}

Note that the converse of this theorem does not hold: Figure \ref{Fig_notTreebased}
shows that the converse of the first part of Theorem \ref{thm-suff} does not hold, because it depicts a treebased network that does not contain a path from one attachment point of a leaf to any other one which visits all inner vertices. Note that such a path would imply a Hamiltonian path from one leaf to another one (when disregarding the remaining leaves), which does not exist. 

Moreover, Figure \ref{leafcutting} shows an example of a treebased network for which $\mathcal{LCUT}(N^u)$ is not Hamilton connected. This implies that the second part of Theorem \ref{thm-suff} cannot be converted. 

Figure \ref{Fig_LCON} shows an example of a treebased network for which there is no $G$ in $\mathcal{LCON}(N^u)$ such that $\mathcal{LCON}(N^u)$ is Hamiltonian. Neither $G_1$, $G_2$ or $G_3$ in $\mathcal{LCON}(N^u)$ contain a Hamiltonian cycle.
So conditions three and four stated by Theorem \ref{thm-suff} are also sufficient, but not necessary. \\

Moreover, before we turn our attention to the proof of the theorem, we want to mention that concerning $\mathcal{LCON}(N^u)$, the exact order in which the leaves are connected can play a fundamental role. Figure \ref{Fig_OrderLcon} shows a phylogenetic network (based on the famous Petersen graph) which is treebased, and two different graphs in $\mathcal{LCON}(N^u)$. However, only one of them is Hamiltonian, while the other one is not, because the Petersen graph is non-Hamiltonian (see for example properties of the Petersen graph in the \enquote{House of graphs} database (graph ID 660); \citet{Brinkmann2013}). 

\begin{figure}[htbp]
	\centering
	\includegraphics[scale=0.2]{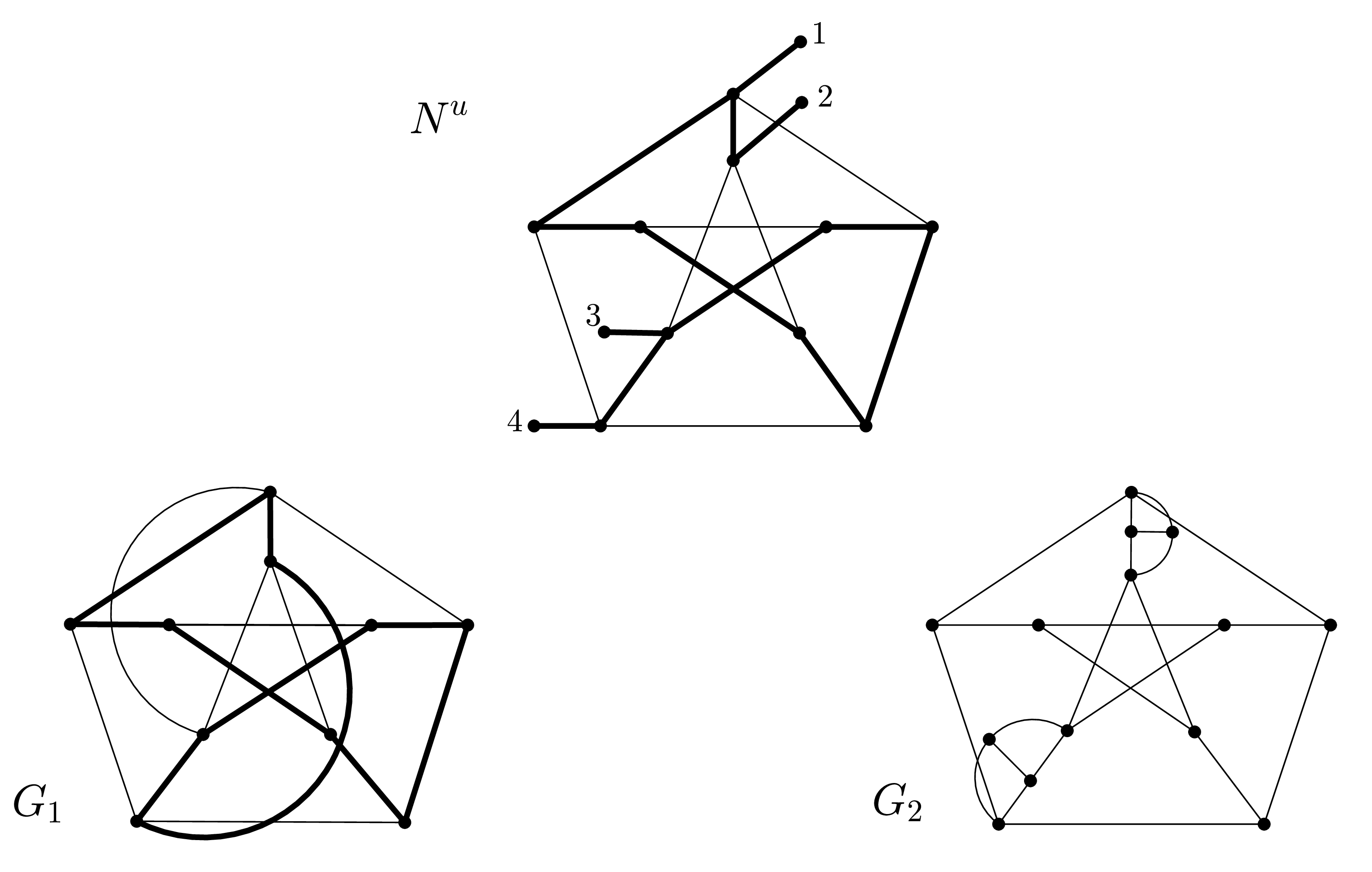}
	\caption{Treebased network $N^u$ (support tree depicted in bold) that is based on the Petersen graph. $G_1$ and $G_2$ are both in $\mathcal{LCON}(N^u)$, but only $G_1$ is Hamiltonian (a Hamiltonian cycle is depicted in bold).}
	\label{Fig_OrderLcon}
\end{figure}

\par \vspace{0.5cm}
We now prove Theorem \ref{thm-suff}.

\begin{proof}[Theorem \ref{thm-suff}] \mbox{}
\begin{enumerate}
\item If $N^u$ contains two leaves $x$ and $y$ with attachment points $u$ and $v$, respectively, such that the edge $\{u,v\}$ is contained in the edge set of $N^u$ and such that there is a path in $N^u$ from $u$ to $v$ visiting all inner vertices of $N^u$, then we can construct a support tree $T$ for $N^u$ as follows: take the path  from $u$ to $v$ visiting all inner vertices of $N^u$ and add all leaves of $N^u$ together with their pending edges to it. As all attachment points of leaves are contained in the path already (as this path visits all inner vertices), now that we re-introduced all leaves, $T$ indeed covers all vertices of $N^u$. As we did not add the edge $\{u,v\}$, there is no cycle. In total, $T$ is a spanning tree of $N^u$. Moreover, its leaf set must coincide with $N^u$: All leaves of $N^u$ are at the same time leaves of $T$ (as a degree-1 vertex of $N^u$ naturally also has degree 1 in $T$). Moreover, all vertices on the path from $u$ to $v$ have degree at least 2, except for $u$ and $v$. But as $u$ and $v$ were attachment points of leaves, after re-attaching those, they now also have degree at least 2 in $T$. So $T$ cannot have any leaves that are not leaves of $N^u$. So $T$ is a support tree of $N^u$ and thus $N^u$ is treebased. 
\item Let $N^u$ be a $\mathcal{H}$-connected network, i.e. let $\mathcal{LCUT}(N^u)$ be Hamilton connected. Consider any two leaves $x$ and $y$ of $N^u$ and their respective attachment points, say $u$ and $v$. As $\mathcal{LCUT}(N^u)$ is Hamilton connected, there is a Hamiltonian path from $u$ to $v$ in $\mathcal{LCUT}(N^u)$. We now consider this path in  $N^u$ and extend it by all pending edges of all leaves. This leads to a tree $T$ that covers all inner vertices on the original path from $u$ to $v$ and all leaves as they were re-attached. There indeed cannot be any cycles as the Hamiltonian path itself has no cycle and as adding leaves, which are of degree 1, cannot create cycles. Thus, $T$ is a spanning tree of $N^u$. Moreover, the leaf set of $T$ coincides with that of $N^u$: All vertices on the path from $u$ to $v$ except for $u$ and $v$ have degree 2 before re-attaching their leaves. $u$ and $v$ have degree 1 in the path, but their leaves $x$ and $y$ also got re-attached, so in the final tree, they have degree 2. Therefore, the only degree-1 vertices in $T$ are the leaves of $N^u$. Thus, $T$ is a support tree and $N^u$ is treebased. 
\item Now let us assume that there is a $G$ in $\mathcal{LCON}(N^u)$ such that $G$ contains a Hamiltonian cycle that uses at least one of the edges that $N^u$ does not contain, i.e. that were introduced when turning $N^u$ into $G$. Consider such a graph $G$ and such a Hamiltonian cycle. Note that as this cycle covers all vertices of $G$, it in particular covers all vertices to which the leaves of $N^u$ are attached. Moreover, it also covers all vertices of $G$ that are not in $N^u$, namely precisely the vertices of type $a$ and $b$ that may have been added during the construction of $G$ to prevent parallel edges. We will now turn this cycle into a support tree of $N^u$ as follows. 
\begin{itemize}
\item If no new vertices were added when $G$ was constructed, this particularly implies that no connection of leaves led to parallel edges. However, as $N^u$ has at least two leaves, at least one edge of $G$ is not an edge of $N^u$. By assumption, such an edge $\{u,v\}$ is covered by the Hamiltonian cycle of $G$ which we are considering. Then we consider the same cycle in $N^u$, but break the edge $\{u,v\}$ in order to get an acyclic tree. This path tree has only two vertices of degree 1, namely $u$ and $v$. But as the edge $\{u,v\}$ was added during the construction of $G$, it must imply that both $u$ and $v$ are leaf attachment points in $N^u$. Now we re-attach all leaves in order to turn this path tree into a tree $T$ whose only leaves are the leaves of $N^u$ (because now the degrees of both $u$ and $v$ are at least 2) and which by construction covers all vertices of $N^u$. So $T$ is a support tree of $N^u$ and therefore $N^u$ is treebased.

\item If there is a pair of vertices $a$ and $b$ that had to be added to $G$ when it was constructed in order to prevent parallel edges between, say, $u$ and $v$, we continue with constructing a support tree $T$ as follows: First, all edges of the cycle in $G$ that were already present in $N^u$ are considered. Moreover, except for one fixed pair $a$ and $b$ that was added to prevent parallel edges, for all other such pairs $a'$, $b'$ between vertices $u'$ and $v'$, say, as we do not have edges $\{u',a'\}$, $\{a',b'\}$ and $\{b',v'\}$ in $N^u$, we remove them. (Note that up to permuting the names of $u'$ and $v'$, these edges must be contained in the Hamiltonian cycle as otherwise it could not cover $a'$ and $b'$). Instead add to $T$ the corresponding edge $\{u',v'\}$, which must be contained in $N^u$ as otherwise, $a'$ and $b'$ would not have been added during the construction of $G$. Also, if the number of leaves of $N^u$ is odd (after a potential pre-processing step), this means that during the construction of $G$, there may have been another vertex added, say $a''$, for the last leaf $x$ with attachment point $w$, again to prevent parallel edges between, say, $u''$ and $v''$. If this is the case, we must have edges $\{u'',v''\}$, $\{x,w\}$, $\{u'',w\}$ and $\{w,v''\}$ in $N^u$. Note that $G$ does not contain $\{x,w\}$ and $\{u'',v''\}$, but instead $\{w,a''\}$, $\{u'',a''\}$, $\{a'',v''\}$. In order to cover $a''$ and $w$, the Hamiltonian cycle we consider must contain the edge $\{w,a''\}$ and either the pair $\{u'',a''\}$ and $\{w,v''\}$, or the pair $\{v'',a''\}$ and $\{w,u''\}$. In either case, $u''$ and $v''$ are covered by the Hamiltonian cycle in $G$ such that one path between them visits only $a''$ and $b''$, while the other one covers all other vertices of $G$. So for $T$, we keep edge $\{u'',v''\}$ as a replacement for the path containing $a''$ and $b''$, and additionally add edges $\{x,w\}$ and $\{u'',w\}$ in order to re-attach leaf $x$. 
Subsequently, we re-attach all other leaves of $N^u$.

Last, we have to deal with our fixed pair $a$ and $b$. As before, these two vertices can only be covered by the Hamiltonian cycle of $G$ if $u$ and $v$ are connected via one path visiting all vertices of $G$ except $u$ and $v$, and by one path only using $a$ and $b$. However, the existence of $a$ and $b$ implies that there is an edge $\{u,v\}$ in $N^u$. For $T$, we do \emph{not} consider this edge, i.e. we do not translate it from the Hamiltonian cycle of $G$ into $N^u$. This way, when we delete $a$ and $b$ (which we have to as they are not present in $N^u$), $u$ and $v$ will be connected via one path visiting all inner vertices of $N^u$, but as the edge $\{u,v\}$ is not contained in $T$, $T$ is acyclic. Morever, by construction, $T$ covers all vertices of $N^u$. As it was created from a Hamiltonian cycle, it is clear that all vertices along this cycle have degree at least 2 in $T$, except for $u$ and $v$, which is where we broke the cycle. However, as $u$ and $v$ are attachment points of leaves, they have degree at least 2 in $T$ as well. So in total, all inner vertices of $N^u$ are inner vertices of $T$, too, and so $T$ is a support tree of $N^u$ and $N^u$ is treebased. 
\end{itemize}
\item Now assume $G \in \mathcal{LCON}(N^u)$ is Hamiltonian and that $G$ contains two vertices $a$ and $b$ which were added when joining two leaf attachment points $u$ and $v$ when constructing $G$ from $N^u$. As we have seen before, in order to cover $a$ and $b$, the Hamiltonian cycle must contain a path from $u$ to $v$ only visiting $a$ and $b$ (and another path from $u$ to $v$ visiting all other vertices of $G$). So the edge $\{a,b\}$ must be used. That $N^u$ is treebased now follows from Part 3 of this theorem.
\end{enumerate}
This completes the proof.
\end{proof}

We are now in the position to show that some classes of phylogenetic networks are treebased due to well-known graph theoretical properties.

\begin{corollary} \label{10tough} Let $N^u$ be a proper unrooted phylogenetic network with at least two leaves and such that $\mathcal{LCUT}(N^u)$ is not Hamiltonian and such that there is a graph $G$ in $\mathcal{LCON}(N^u)$ which is a 10-tough chordal graph. Then, $N^u$ is treebased.
\end{corollary}

\begin{proof} We know from \citet{Kabela2017} that every 10-tough chordal graph is Hamiltonian. Thus, $G$ is Hamiltonian. However, as $\mathcal{LCUT}(N^u)$ is not Hamiltonian, the cycle in $G$ must use edges that are not contained in $N^u$. Thus, $N^u$ is treebased by Theorem \ref{thm-suff}, Part 3. This completes the proof.
\end{proof}

Note that while Corollary \ref{10tough} implies a connection between chordal graphs and treebasedness, not all chordal graphs are treebased. This can be seen in Figure \ref{Fig_Chordal_not_treebased}. However, we will now prove that such a scenario cannot happen when $N^u$ is binary.

\begin{figure}[htbp]
	\centering
	\includegraphics[scale=0.2]{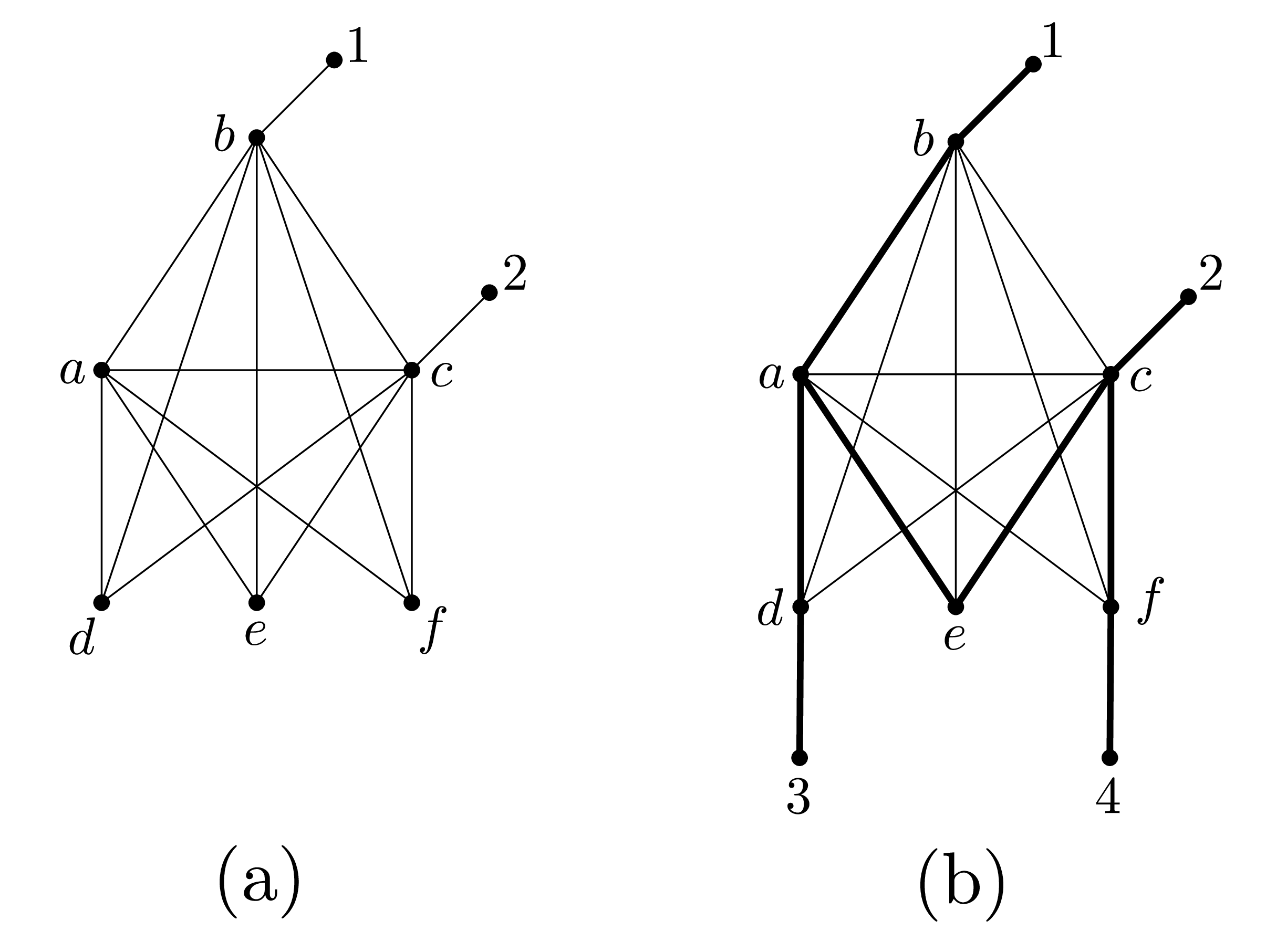}
	\caption{(a) Chordal graph that -- considered as an unrooted non-binary phylogenetic network -- is not treebased, because there is no Hamiltonian path between $1$ and $2$. (b) Attaching at least two more leaves to either $d$, $e$ or $f$ produces a treebased network (the support tree is depicted in bold).}
	\label{Fig_Chordal_not_treebased}
\end{figure}

 \begin{theorem} \label{binarychordal} 
Let $N^u$ be a proper unrooted phylogenetic network with at least two leaves. Then, if $N^u$ is binary and chordal, $N^u$ is edgebased (and thus, by Theorem \ref{edgebased_implies_treebased}, also treebased).
 \end{theorem}

\begin{proof}
Let $N^u$ be a proper unrooted phylogenetic network with at least two leaves, such that $N^u$ is binary and chordal. 
If $N^u$ is a tree, there is nothing to show, because $N^u$ is trivially edgebased and treebased. 
Thus, now assume that $N^u$ is not a tree. This implies that $N^u$ has to contain at least one non-trivial blob (if it only contained trivial blobs, $N^u$ would be a tree).

By Lemma \ref{blobreduce}, it now suffices to consider one such non-trivial blob of $N^u$, which we denote by $G$. As $G$ is a non-trivial blob, this means that $G$ has no cut edges and no leaves; in particular, $G$ only has vertices of degree 2 and 3, and as $N^u$ has leaves, the existence of a degree-2 vertex $u$ in $G$ is guaranteed. Moreover, $G$ is still chordal (as the deletion of leaves cannot destroy chordality). Now note that in the given chordal graph, \emph{every} vertex belongs to a triangle by Lemma \ref{chordaltriangle} in the Appendix. Therefore, this applies also to $u$, so $u$ and its neighbors $v$ and $w$ form a triangle. 

So we have a chordal graph, in which all vertices have degree at least 2 and at most 3, and we have one vertex $u$ of degree 2, which belongs to a triangle $uvw$. We now repeat the following procedure: 

First we suppress $u$. As $v$ and $w$ are adjacent (they belong to the triangle $uvw$), this causes a parallel edge $e=\{v,w\}$. Deleting this parallel edge will strictly decrease the degrees of $v$ and $w$. 
Thus, if the degrees of $v$ and $w$ were both 2 before the deletion of the parallel edge, we now get two new leaves. However, in this case, the edge $e=\{v,w\}$ is the only edge left, and thus, $N^u$ is edgebased.
Else, if either $v$ or $w$ (or both) have degree 2 after the deletion of the parallel edge, we re-name this vertex (or one of them, respectively) to $u$. Again, as the current graph is still chordal (we did not increase the cycle length of any cycle), the new vertex $u$ of degree 2 belongs to a triangle, whose suppression causes a parallel edge and so forth. So we can repeat this procedure, which is depicted by Figure \ref{Fig_BinaryChordal}, until only one edge remains. This completes the proof.

\begin{figure}
\centering
\includegraphics[scale=0.25]{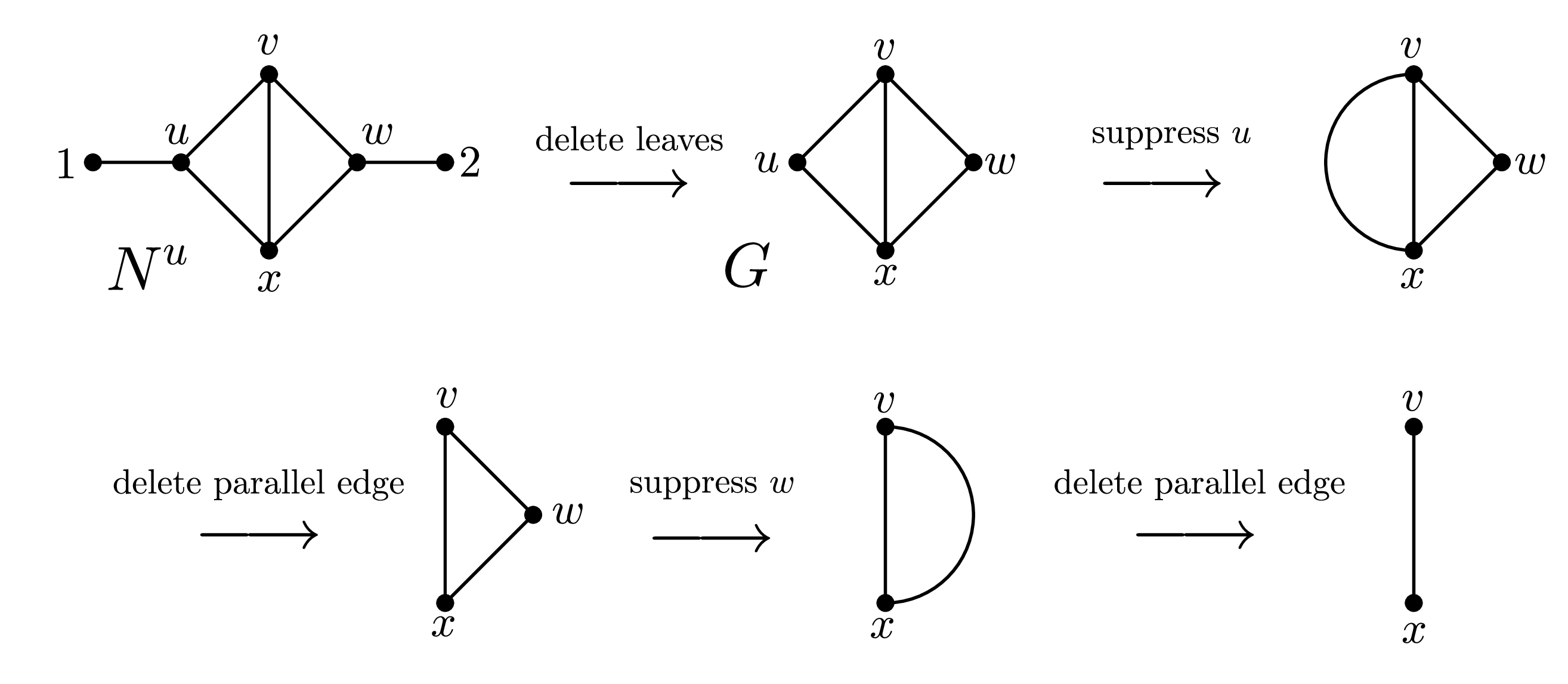}
\caption{Proper unrooted phylogenetic network $N^u$ (consisting of one non-trivial blob and two trivial blobs (leaves)) that is binary and chordal. After deleting its leaves, it can be reduced to a single edge by a sequence of vertex suppression and edge deletion operations. First, we consider the triangle $uvx$ and suppress $u$. This results in a parallel edge between $v$ and $x$, which gets deleted. Then, the triangle $vwx$ is considered and $w$ is suppressed. Deleting the resulting parallel edge between $v$ and $x$ leads to a single edge. This implies that $N^u$ is edgebased.}
\label{Fig_BinaryChordal}
\end{figure}
\end{proof}

\begin{remark} \label{remark_perfect}
A generalization of chordal graphs are so-called \emph{perfect} graphs (also known as \emph{Berge} graphs), where a perfect graph is a graph $G$ such that neither $G$ nor its complement $\bar{G}$ contain an odd cycle of length $\geq 5$. An interesting question is whether the result of all binary chordal networks being edgebased (Theorem \ref{binarychordal}) generalizes to binary perfect networks. If we only consider $\mathcal{LCUT}(N^u)$ this is not necessarily the case, as there are networks $N^u$ such that $\mathcal{LCUT}(N^u)$ is binary and perfect, but not edgebased (cf. Figure \ref{Fig_perfectbinary}).
\end{remark}

\begin{figure}[htbp]
	\centering
	\includegraphics[scale=0.5]{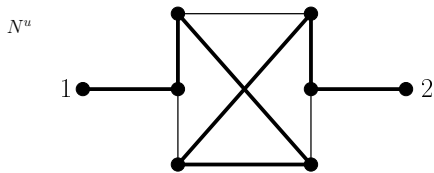}
	\caption{Proper phylogenetic network such that $\mathcal{LCUT}(N^u)$ is a binary and perfect graph. $N^u$ is treebased (the support tree is highlighted in bold), but not edgebased.}
	\label{Fig_perfectbinary}
\end{figure}

\subsubsection{Relationships between different classes of treebased networks}
In the previous sections we have introduced a variety of networks that are guaranteed to be treebased, ranging from edgebased networks to $\mathcal{H}$-connected ones. We conclude this section by analyzing the relationships between these classes.

Figure \ref{venndiagram} shows a Venn diagram of different classes of proper phylogenetic networks in connection with treebasedness. 

If the intersection of different classes of such networks is non-empty, Figure \ref{venndiagram} contains representative examples. 
To summarize, we have
\begin{itemize}
\item There exist proper phylogenetic networks that are treebased (e.g. Figure 6 in \citet{paper2}).
\item Not all proper phylogenetic networks are treebased (e.g. Figure 7 in \citet{paper2}).
\item All proper phylogenetic networks that are edgebased are treebased (cf. Theorem \ref{edgebased_implies_treebased}).
\item All proper phylogenetic networks that are binary and chordal, are edgebased and thus also treebased (cf. Theorem \ref{binarychordal}).
\item Proper phylogenetic networks that are chordal, are not necessarily treebased (cf. Figure \ref{Fig_Chordal_not_treebased}).
\item Proper phylogenetic networks that are $\mathcal{H}$-connected, are treebased (cf. Theorem \ref{thm-suff}, Part 2). 
\end{itemize}

Note, however, that the intersection of networks that are at the same time edgebased, $\mathcal{H}$-connected, and non-chordal is empty. This is due to the fact that such networks do not exist. We will explain this subsequently (cf. Remark \ref{remark_emptyset}).
Moreover, even if the network \emph{is} chordal, the classes of $\mathcal{H}$-connected and edgebased networks have only a very small overlap, as we will show in the following (cf. Theorem \ref{emptyset}).

So they are indeed very different types of networks. We will subsequently fully characterize their overlap, i.e. we will describe which phylogenetic networks are $\mathcal{H}$-connected and edgebased. In particular, we will show that they are all chordal. We start with the following theorem. 

\begin{figure}[H]
\centering
\includegraphics[scale=0.65]{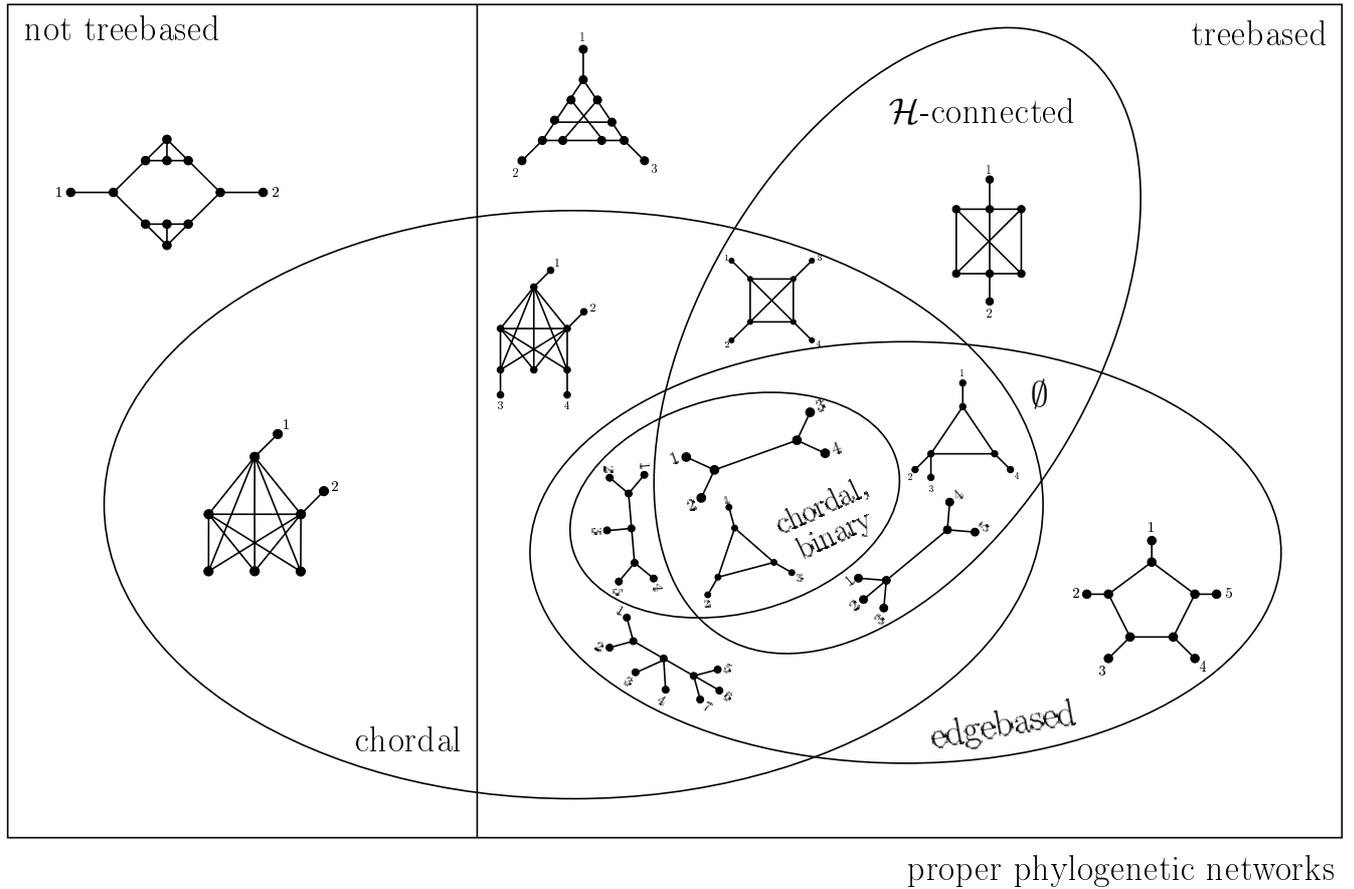}
\caption{Venn diagram of different classes of proper phylogenetic networks and their connection to treebasedness.}
\label{venndiagram}
\end{figure}

\begin{theorem} \label{emptyset}Let $N^u$ be an edgebased and $\mathcal{H}$-connected phylogenetic network. Then, $\mathcal{LCUT}(N^u)$ contains less than four vertices.
\end{theorem}

\begin{remark} \label{remark_emptyset}
This theorem actually shows that there are no edgebased, $\mathcal{H}$-connected and non-chordal phylogenetic networks, because non-chordal networks require a cycle of length at least four (without a chord) and thus at least four vertices in $\mathcal{LCUT}$. 
\end{remark}

Before we can prove Theorem \ref{emptyset}, we need two more lemmas.

\begin{lemma}\label{Hnocuts}
Let $N^u$ be an $\mathcal{H}$-connected phylogenetic network such that $\mathcal{LCUT}(N^u)$ consists of more than just one edge. Then, $\mathcal{LCUT}(N^u)$ contains no cut vertices and no cut edges.
\end{lemma}

\begin{proof} Let $N^u$ be an $\mathcal{H}$-connected phylogenetic network such that $\mathcal{LCUT}(N^u)$ consists of more than just one edge.  Assume $\mathcal{LCUT}(N^u)$ contains a cut vertex $v$. As $v$ is a cut vertex, there are at least two more vertices $u$ and $w$ which get disconnected by the removal of $v$. So the only paths from $u$ to $w$ in $\mathcal{LCUT}(N^u)$ are all via $v$. This implies that there cannot be a Hamiltonian path from $u$ to $v$, because any sequence of vertices starting at $u$ and going via $w$ (and possibly other vertices) to $v$ would visit $v$ at least twice. Thus, if $N^u$ contains cut vertices, $N^u$ is not $\mathcal{H}$-connected, which is a contradiction.

On the other hand, if $\mathcal{LCUT}(N^u)$ contains a cut edge $e=\{u,v\}$, this implies that $u$ and $v$ are cut vertices. So again, this leads to a contradiction. This completes the proof. 
\end{proof}

\begin{lemma} \label{nodeg2}
Let $G=(V,E)$ be a Hamilton-connected graph with at least 4 vertices. Then for all $v \in V$ we have $deg(v)>2$. 
\end{lemma}

\begin{proof} First note that in a Hamilton-connected graph, there are clearly no isolated vertices, i.e. $deg(v)>0$ for all $v \in V$. Moreover, there cannot be any vertices of degree 1 in $G$, because by the same arguments used in the proof of Lemma \ref{Hnocuts}, $G$ cannot contain a cut edge (but each edge incident to a leaf would be a cut edge). So we conclude $deg(v)>1$ for all $v \in V$. 
So now let $u$, $v$, $w$ be in $V$ such that $deg(v)=2$ and $u$ and $w$ are the two neighbors of $v$ in $G$, and let $x$ denote some other vertex in $V$, which must exist as $|V|\geq 4$. 
Now there is no Hamiltonian path from $u$ to $w$ visiting both $v$ and $x$.
If a path from $u$ to $w$ starts via $v$, $x$ cannot be contained in it unless either $u$ or $w$ is visited twice. On the other hand, if a path from $u$ to $w$ visits $x$ before $v$, then $v$ can only be reached by visiting either $u$ or $w$ twice. In both cases, the corresponding path from $u$ to $w$ is not Hamiltonian. 
So this is a contradiction, as $G$ is Hamilton-connected. This completes the proof. 
\end{proof}

We are now in the position to prove Theorem \ref{emptyset}.

\begin{proof}[Theorem \ref{emptyset}]Assume the statement is wrong, i.e. assume there exists a phylogenetic network $N^u$ which is $\mathcal{H}$-connected and edgebased and such that  $\mathcal{LCUT}(N^u)$ contains at least four vertices.  As $N^u$ is $\mathcal{H}$-connected, by Lemma \ref{nodeg2}, $\mathcal{LCUT}(N^u)$ contains no vertices of degree at most 2 because it by assumption contains at least four vertices. We now consider $\mathcal{LS}(N^u)$. When we create $\mathcal{LS}(N^u)$ from $\mathcal{LCUT}(N^u)$ (note that we can go from $N^u$ to $\mathcal{LS}(N^u)$ via $\mathcal{LCUT}(N^u)$ as the order of restriction operations does not matter due to Theorem \ref{thm-edgebased}), there are no degree-2 vertices to suppress. Moreover, there are no parallel edges, because if $\mathcal{LCUT}(N^u)$ contained parallel edges, so would $N^u$, which contradicts the definition of a phylogenetic network. Additionally, there can also be no leaves as this would imply degree-1 vertices (which cannot exist due to Lemma \ref{nodeg2}). So there is no leaf to suppress, no degree-2 vertex and no parallel edge -- in other words, $\mathcal{LS}(N^u)=\mathcal{LCUT}(N^u)$ as there is nothing to shrink. As $|V(\mathcal{LCUT}(N^u))|\geq 4$, this implies $|V(\mathcal{LS}(N^u))|\geq 4$, which shows that $N^u$ cannot be edgebased. This is a contradiction. Therefore, the assumption was wrong and such a network cannot exist. This completes the proof. 
\end{proof}

Next we want to characterize all cases in which a phylogenetic network is $\mathcal{H}$-connected and edgebased. We will show that the number of networks in this class is in fact very small -- in fact we can fully characterize their $\mathcal{LCUT}$ graphs.

\begin{theorem} \label{Hedgebased} Let $N^u$ be an $\mathcal{H}$-connected and edgebased phylogenetic network. Then, one of the following two cases holds:
\begin{itemize}
\item $N^u$ is a tree with at most one inner edge, i.e. $\mathcal{LCUT}(N^u)$ consists of either only one vertex or one edge.
\item $N^u$ contains precisely one cycle, and this cycle is a triangle, and $\mathcal{LCUT}(N^u)$ consists only of this triangle. 
\end{itemize}
In particular, $N^u$ is chordal.
\end{theorem}

\begin{proof} Let $N^u$ be an $\mathcal{H}$-connected and edgebased phylogenetic network. By Theorem \ref{emptyset}, $\mathcal{LCUT}(N^u)$ contains at most 3 vertices. We now distinguish two cases: 

\begin{itemize}
\item If $|V(\mathcal{LCUT}(N^u))| \leq 2$, then $N^u$ is clearly a tree (because the at most two vertices of $\mathcal{LCUT}(N^u)$ cannot form a cycle) with at most one inner edge (because there is at most one edge in $\mathcal{LCUT}(N^u)$ as there are at most two vertices). So in this case, the first case of the theorem holds.
\item Now suppose that $|V(\mathcal{LCUT}(N^u))| =3$. Then we are clearly not in the first case of the theorem. Now assume that the three vertices $u$, $v$ and $w$ of $\mathcal{LCUT}(N^u)$ do not form a cycle. As $\mathcal{LCUT}(N^u)$ is connected, this implies that $u$, $v$ and $w$ form a path, i.e. $\mathcal{LCUT}(N^u)$ contains precisely two edges, say $e_1=\{u,v\}$ and $e_2=\{v,w\}$. Then, both $e_1$ and $e_2$ are cut edges, as their removal would disconnect $u$ and $w$. But as $N^u$ is $\mathcal{H}$-connected, $\mathcal{LCUT}(N^u)$ does not contain any cut edges due to Lemma \ref{Hnocuts}. So this is a contradiction. So the three vertices $u$, $v$ and $w$ must form a triangle. As there cannot be a further vertex in $\mathcal{LCUT}(N^u)$, this completes the proof.
\end{itemize}

\end{proof}

So by Theorem \ref{Hedgebased} we know that {\em all } $\mathcal{H}$-connected and edgebased phylogenetic networks are chordal, and that all of them either have a single vertex, a single edge or a triangle as their $\mathcal{LCUT}$ graph. However, note that the number of networks with these properties is not restricted, because arbitrarily many leaves can be attached to such $\mathcal{LCUT}$ graphs.

\section{Discussion and Conclusion} \label{sec_discussion}
The main aim of the present manuscript was to link treebasedness of phylogenetic networks to classical graph theory. To be precise, we have established links between treebasedness and the theory of Hamiltonian or Hamilton connected graphs as well as between treebasedness and the family of generalized series parallel graphs.

The close links of treebased networks and Hamiltonian or Hamilton connected graphs allow for sufficient criteria that make a network treebased, even if none of these criteria is necessary. We are sure that future research will establish even more links between Hamiltonicity of graphs and treebasedness of phylogenetic networks. Furthermore, as more and more classes of graphs are e.g. discovered to be Hamilton connected (see e.g. \citet{Hu2005, Alspach2013}), this will deliver more and more known graphs that lead to treebased networks.

However, none of these links to Hamiltonicity leads to classes of networks for which treebasedness can be efficiently checked, as the mentioned graph theoretical counterparts of treebasedness, e.g. testing if a graph is Hamiltonian, are known to be NP-complete (cf. \citet{Karp1972}). 

Additionally, however, we have introduced a class of networks that are guaranteed to be treebased, namely the class of edgebased networks. Interestingly, edgebased networks are closely related to another important concept in classic graph theory, namely the class of generalized series parallel graphs. In the present manuscript, we have shown that the links between treebasedness, edgebasedness and generalized series parallel graphs lead to a sufficient criterion for treebasedness that can be verified in linear time. In this regard, edgebased phylogenetic networks form a class of treebased networks that can easily be found. 
We, for example, showed that all unrooted, binary, chordal phylogenetic networks are edgebased. 
As mentioned in Remark \ref{remark_perfect}, an interesting question is whether this result generalizes to other classes of proper phylogenetic networks, e.g. to binary, perfect ones.
It would also be of interest to analyze if edgebased networks frequently occur in practice, i.e. when phylogenetic networks are constructed from biological data. As research on reconstructing phylogenetic networks from data is still at its beginning, this is hard to predict. However, we are sure that edgebased networks will be of practical relevance in future.

We concluded our study with analyzing the relationships between the classes of treebased networks presented in this manuscript and summarized them in Figure \ref{venndiagram}. Again, we are sure that future research will characterize more and more classes of treebased networks, adding to our results.

\section*{Acknowledgement} We wish to thank Clemens A. Fischer for helpful discussions concerning chordal graphs. Moreover, we thank two anonymous reviewers for their helpful comments on an earlier version of this manuscript.
The third and fifth author also thank the state Mecklenburg-Western Pomerania for the Landesgraduierten-Studentship. Moreover, the second author thanks the University of Greifswald for the Bogislaw-Studentship and the fifth author thanks the German Academic Scholarship Foundation for a studentship.

\newpage
\bibliographystyle{spbasic}      
\bibliography{References1}   

\section{Appendix}

\setcounter{lemma}{2}
\begin{lemma}
Let $G$ be a connected graph with vertex set $V(G)$ and edge set $E(G)$. Let $G'$ result from $G$ by deleting one loop. Then, $K_2$ is a restricted topological subgraph of $G$ if and only if $K_2$ is a restricted topological subgraph of $G'$.
\end{lemma}

\begin{proof} By Lemma \ref{addsth}, if $K_2$ is a restricted topological subgraph of $G'$, then it is also a restricted topological subgraph of $G$, so this direction is clear. 

Now assume that there is a graph $G$ such that $K_2$ is a restricted topological subgraph of $G$, but if we delete one loop of $G$ to obtain $G'$, $K_2$ no longer is a restricted topological subgraph. If such graphs exist, we may consider a minimal one in terms of the number of edges. So assume that $G$ is minimal with this property, i.e. for all graphs with fewer edges we know that if $K_2$ is a restricted topological subgraph, this property still holds after the deletion of a loop.

As $G$ has $K_2$ as a restricted topological subgraph, there is a sequence of the restriction operations which converts $G$ into $K_2$. However, there is also a loop $\{u,u\}$, whose deletion turns $G$ into $G'$. So the first operation to convert $G$ into $K_2$ cannot be the deletion of this loop. So the first step is either the deletion of a leaf (together with its incident edge), the suppression of a degree-2 vertex (which `melts' two edges into one), the deletion of one copy of a parallel edge or the deletion of some other loop. In all cases, we obtain a graph $G''$ which has fewer edges than $G$ and which must have $K_2$ as a restricted topological subgraph as it is on the path from $G$ to $K_2$. However, as $G$ is minimal with the property that the deletion of a loop can cause a loss of $K_2$ as a restricted topological subgraph, we can delete the loop $\{u,u\}$ from $G''$ to obtain $\widetilde{G}$, which again has $K_2$ as a restricted topological subgraph. By Lemma \ref{addsth}, we can undo the first step which we took from $G$ to $G''$, i.e. we can re-add the deleted leaf or degree-2 vertex or parallel edge or loop (note that this means we convert $\widetilde{G}$ into $G'$), without losing the property that $K_2$ is a restricted topological subgraph. Thus, $K_2$ is a restricted topological subgraph of $G'$, which contradicts our assumption. Therefore, such graphs cannot exist, which means that the question whether $K_2$ is a restricted topological subgraph of a graph $G$ cannot depend on the loops of $G$. This completes the proof.
\end{proof}

\begin{lemma}
Let $G$ be a connected graph with vertex set $V(G)$ and edge set $E(G)$. Let $G'$ result from $G$ by deleting one copy of a parallel edge. Then, $K_2$ is a restricted topological subgraph of $G$ if and only if $K_2$ is a restricted topological subgraph of $G'$.
\end{lemma}

\begin{proof} By Lemma \ref{addsth}, if $K_2$ is a restricted topological subgraph of $G'$, then it is also a restricted topological subgraph of $G$, so this direction is clear. 

Now assume that there is a graph $G$ such that $K_2$ is a restricted topological subgraph of $G$, but if we delete a copy of a parallel edge of $G$ to obtain $G'$, $K_2$ no longer is a restricted topological subgraph. If such graphs exist, we may consider a minimal one in terms of the number of edges. So assume that $G$ is minimal with this property, i.e. for all graphs with fewer edges we know that if $K_2$ is a restricted topological subgraph, this property still holds after the deletion of a parallel edge.

As $G$ has $K_2$ as a restricted topological subgraph, there is a sequence of the restriction operations which converts $G$ into $K_2$. However, there is also an edge $e$ of which multiple copies exist, such that the deletion of $e$ turns $G$ into $G'$. So the first operation to convert $G$ into $K_2$ cannot be the deletion of $e$. Thus, the first step is either the deletion of a leaf (together with its incident edge), the suppression of a degree-2 vertex (which `melts' two edges into one), the deletion of one copy of a parallel edge other than $e$ or the deletion of a loop. In all cases, we obtain a graph $G''$ which has fewer edges than $G$ and which must have $K_2$ as a restricted topological subgraph as it is on the path from $G$ to $K_2$. However, as $G$ is minimal with the property that the deletion of a parallel edge can cause a loss of $K_2$ as a restricted topological subgraph, if $e$ is contained in $G''$, we can delete one copy of $e$ from $G''$ to obtain $\widetilde{G}$, which again has $K_2$ as a restricted topological subgraph. By Lemma \ref{addsth}, we can now undo the first step which we took from $G$ to $G''$, i.e. we can re-add the deleted leaf or degree-2 vertex or parallel edge or loop (note that this means we convert $\widetilde{G}$ into $G'$) to $\widetilde{G}$, without losing the property that $K_2$ is a restricted topological subgraph. Thus, $K_2$ is a restricted topological subgraph of $G'$, which contradicts our assumption. 

If, on the other hand, $G''$ does not contain $e$, this must imply that $e$ disappeared on the way from $G$ to $G''$ by one of the other operations. Note that a leaf deletion only affects a degree-1 vertex and its incident edge, which thus cannot be a parallel edge (else the vertex would have degree at least 2). Moreover, the deletion of a loop -- even if it was parallel, i.e. even if it existed multiple times -- would not cause the disappearance of an edge $e$ which is present multiple times in $G$; and neither would the deletion of another parallel edge which has nothing to do with $e$. So the only way for $e$ to disappear in the first step is if there are precisely two copies of $e=\{u,v\}$ that lead to a vertex $v$ which is \emph{only} incident to these two edges $e$, i.e. $deg(v)=2$. Then, the suppression of $v$ would lead to a loop $\{u,u\}$, and indeed no copy of $e$ would be present in $G''$. However, in this case, we know by Lemma \ref{loopdelete} that we can delete loop $\{u,u\}$ to obtain $G'''$, and $G'''$ still has $K_2$ as a restricted topological subgraph. As above, we can now undo the first step (which we took from $G$ to $G''$) in $G'''$ by Lemma \ref{addsth}. This leads to a graph $\widetilde{G}$, which still has $K_2$ as a restricted topological subgraph. Again by Lemma \ref{addsth}, we can then add vertex $v$ and connect it to vertex $u$ with one new edge $e=\{u,v\}$. This is equivalent to introducing a new leaf and thus keeps $K_2$ as a restricted topological subgraph. However, the resulting graph is $G'$, which means it cannot have $K_2$ as a restricted topological subgraph by assumption. Therefore, this is a contradiction.  

So in both cases we obtain a contradiction, which means that such graphs cannot exist. So the question whether $K_2$ is a restricted topological subgraph of a graph $G$ cannot depend on $G$'s copies of multiple edges. This completes the proof.
\end{proof}

\setcounter{lemma}{4}
\begin{lemma}
Let $G=(V,E)$ be a simple and biconnected $SP$ graph with at least three vertices. Then, there exists a spanning tree $T$ in $G$ whose leaves correspond to degree-2 vertices of $G$. In particular, no vertex $v \in V$ of $G$ with $deg(v) > 2$ is a leaf in $T$. 
\end{lemma}

Recall that we call such a spanning tree a \emph{valid} spanning tree (cf. Remark \ref{remark_valid}).
In order to prove Lemma \ref{Lemma_SP_SpanningTrees}, we require the following lemma from \citet{Song2007}. Therefore, recall that $N(v)$ denotes the neighborhood of a vertex $v$ in $G$, i.e. the set of vertices adjacent to $v$.

\setcounter{lemma}{7}
\begin{lemma}[adapted from Song 2007] \label{Lemma_SP_Cases}
Let $G=(V,E)$ be a simple and biconnected $SP$ graph with $\vert V \vert \geq 5$. Then one of the following conditions holds:
	\begin{enumerate}
	\item $G$ has two adjacent degree-2 vertices $x$ and $y$;
	\item $G$ has two different degree-2 vertices $x$ and $y$ and $N(x)=N(y)$;
	\item $G$ has a degree-4 vertex $z$ adjacent to two degree-2 vertices $x$ and $y$ such that $N(z) \setminus \{x,y\} = \{N(x) \cup N(y) \} \setminus \{z\}$;
	\item $G$ has a degree-3 vertex $w$ with $N(w)=\{x,y,z\}$ such that both $x$ and $y$ are degree-2 vertices, $N(x)=\{z,w\}$ and edge $\{y,z\} \notin E$;
	\item $G$ has two adjacent degree-3 vertices $x$ and $y$ such that $N(x) \cap N(y) = \{z\}$ and $N(z) = \{x,y\}$;
	\item $G$ has two adjacent degree-3 vertices $w_1$ and $w_2$ such that $N(w_1) = \{x,z_1, w_2\}$, $N(w_2) = \{y, z_2, w_1\}$, $N(x) = \{z_1, w_1\}$ and $N(y) = \{z_2, w_2\}$;
	\item $G$ has a degree-3 vertex $w$ with $N(w)=\{x,y,z\}$ such that $N(z)=\{w,y\}$ and edge $\{x,y\} \in E$;
	\item $G$ has two non-adjacent degree-3 vertices $w_1$ and $w_2$ such that $N(w_1) = \{x,y,z_1\}$, $N(w_2) = \{x,y,z_2\}$, $N(z_1)=\{x,w_1\}$ and $N(z_2)=\{y,w_2\}$;
	\item $G$ has two non-adjacent degree-3 vertices $w_1$ and $w_2$ such that $N(w_1) = \{x,y,z_1\}$, $N(w_2) = \{x,y,z_2\}$, $N(z_1) = \{x,w_1\}$ and $N(z_2) = \{x,w_2\}$;
	\item $G$ has a degree-3 vertex $w$ with $N(w)=\{x,z_1,z_2\}$ such that there is a degree-2 vertex $y \in N(z_1) \cap N(z_2)$ and $N(x) = \{z_1, w\}$.
	\end{enumerate}
\end{lemma}

Based on this we can now prove Lemma \ref{Lemma_SP_SpanningTrees}.
\begin{proof}[Proof of Lemma \ref{Lemma_SP_SpanningTrees}]
We prove this statement by induction on the number $n \coloneqq |V|$ of vertices of $G$. 
For $n=3, \ldots, 6$ we generated a catalog of all simple and biconnected $SP$ graphs with $n$ vertices as follows: We retrieved all simple and connected graphs with $n=3,\ldots,6$ vertices from the \enquote{House of graphs} database \citep{Brinkmann2013} and filtered them for biconnected $SP$ graphs with the computer algebra system Mathematica \citep{Mathematica}. First, each downloaded graph $G$ was checked for biconnectedness with the Mathematica function \texttt{KVertexConnectedGraphQ[$G,2$]}. For each of the remaining graphs it was then checked whether a reduction to $K_2$ via the series reduction rules (cf. p.~\pageref{sp}) was possible. If not, the respective graph was discarded. The remaining simple and biconnected $SP$ graphs are depicted in Figure \ref{Fig_Catalog}. We exhaustively analyzed all these graphs to show that in all cases, there exists a valid spanning tree having only degree-2 vertices of the respective $SP$ graph as leaves (which is also depicted in Figure \ref{Fig_Catalog}). This completes the base case of the induction.

Now, assume that the statement holds for all simple and biconnected $SP$ graphs with up to $n-1$ vertices and let $G=(V,E)$ be a simple and biconnected $SP$ graph with $n \geq 7$ vertices. 

According to Lemma \ref{Lemma_SP_Cases}, we can distinguish between ten cases:

\begin{enumerate}
\item $G$ has two adjacent degree-2 vertices $x$ and $y$ (as depicted in Figure \ref{Fig_Case1}): \\
Let $x' \neq y$ be the second vertex adjacent to $x$ and let $y' \neq x$ be the second vertex adjacent to $y$. Note that as $n \geq 7$, we cannot have $y'=x'$, because in this case $x'$ would be a cut vertex, contradicting the fact that $G$ is biconnected. We now construct a simple and biconnected $SP$ graph $G'$ with $n-1$ vertices from $G$ by suppressing vertex $x$ (cf. Figure \ref{Fig_Case1}). Now, by the inductive hypothesis (as $G'$ is a simple and biconnected graph on $6 \leq n-1 <n$ vertices), there exists a valid spanning tree $T'$ in $G'$. We can now obtain a valid spanning tree $T$ for $G$ as follows:
	\begin{itemize}
	\item If edge $\{x',y\}  \in E(T')$ (note that $\{x',y\} \in E(G') \setminus E(G)$), we replace this edge by $\{x',x\}$ and $\{x,y\}$ to obtain $T$. Note that in this case $x$ is not a leaf of $T$.
	\item If edge $\{x',y\} \notin E(T')$, we add either one of the edges $\{x',x\}$ or $\{y,x\}$ to $T'$ in order to obtain $T$. This implies that $x$ is a leaf in $T$, but since $x$ was a degree-2 vertex in $G$, this is valid.
	\end{itemize}
In both cases $T$ is a valid spanning tree for $G$. 

	\begin{figure}[htbp]
	\centering
	\includegraphics[scale=0.3]{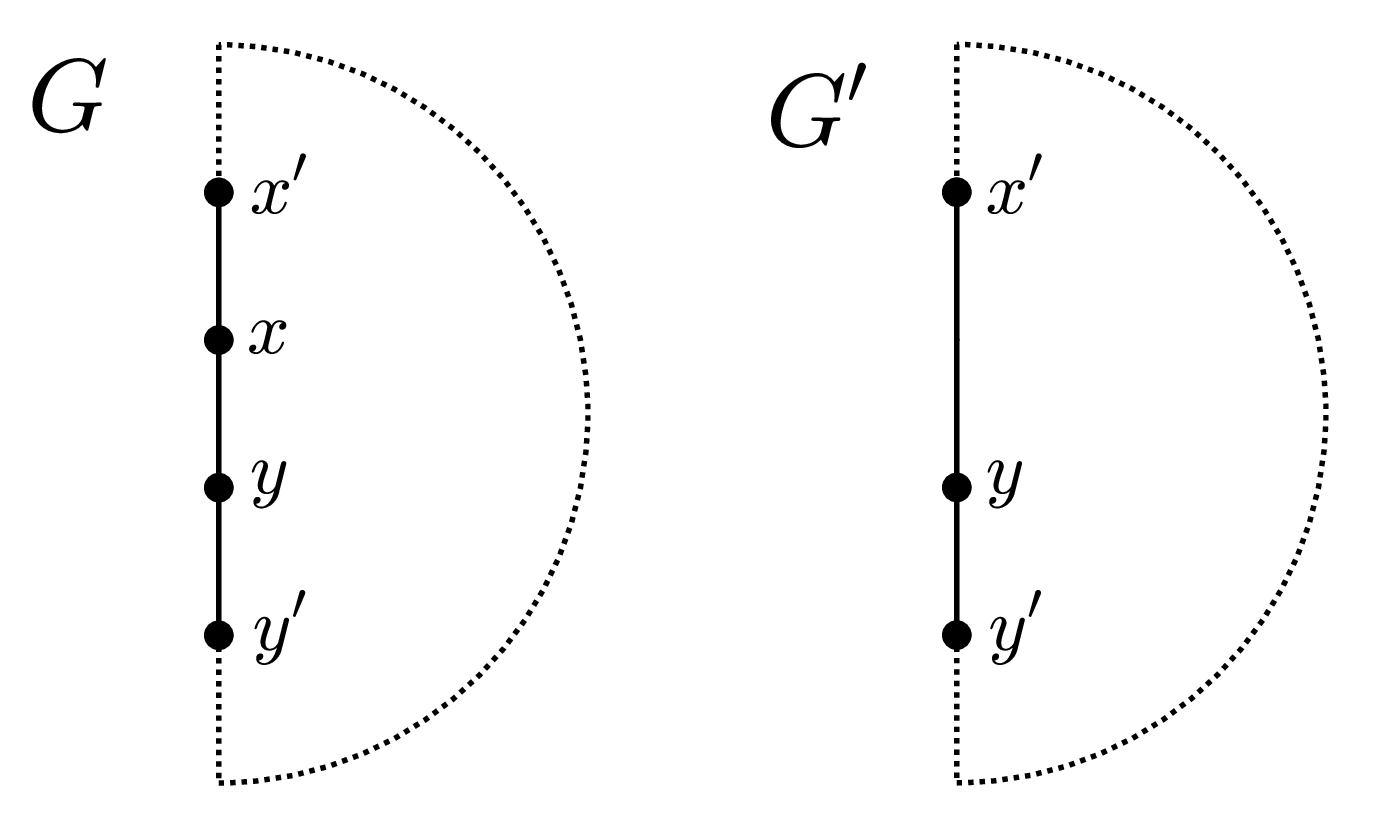}
	\caption{Case 1 in the proof of Lemma \ref{Lemma_SP_SpanningTrees}. The dotted lines depict some path between $x'$ and $y'$ that must exist since $G$ is biconnected.}
	\label{Fig_Case1}
	\end{figure}


\item $G$ has two different degree-2 vertices $x$ and $y$ and $N(x)=N(y)$ (as depicted in Figure \ref{Fig_Case2}): \\
Let $N(x)=N(y)=\{u,v\}$.
We now construct a simple and biconnected $SP$ graph $G'$ with $n-2$ vertices from $G$ by suppressing vertices $x$ and $y$ and deleting all but one copy of the resulting parallel edge $\{u,v\}$ (cf. Figure \ref{Fig_Case2}).

	\begin{figure}[htbp]
	\centering
	\includegraphics[scale=0.3]{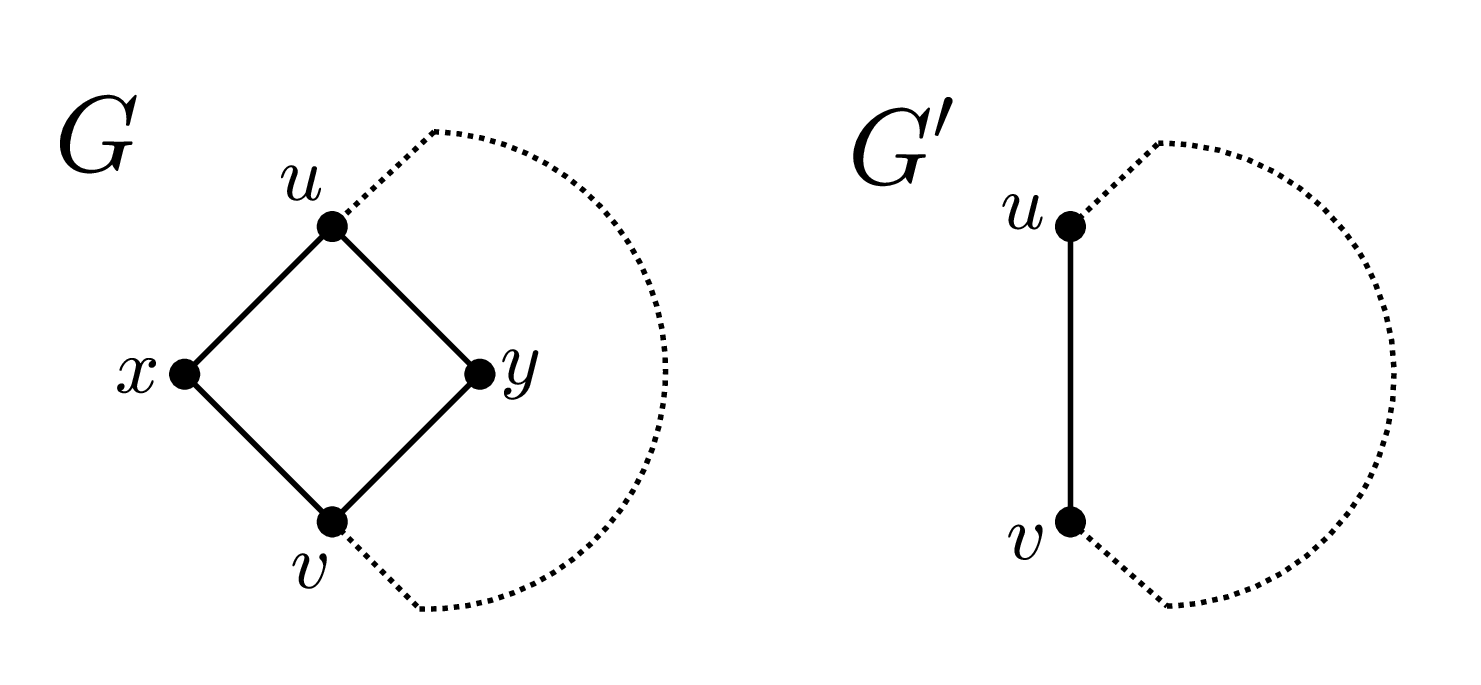}
	\caption{Case 2 in the proof of Lemma \ref{Lemma_SP_SpanningTrees} for $n \geq 7$. The dotted line depicts some path between $u$ and $v$ (possibly consisting of a single edge) that must exist since $G$ is biconnected and $n \geq 7$.}
	\label{Fig_Case2}
	\end{figure}

As $G'$ is a simple and biconnected graph with $5 \leq n-2 < n$ vertices, by the inductive hypothesis, $G'$ has a valid spanning tree $T'$. Note that $u$ and $v$ are potentially leaves in $T'$ (as they are potentially degree-2 vertices in $G'$). We can now construct a valid spanning tree $T$ for $G$ by adding edges $\{u,x\}$ and $\{v,y\}$ (or $\{u,y\}$ and $\{v,x\}$) to $T'$. This guarantees that $u$ and $v$ are internal vertices of $T$, while $x$ and $y$ are leaves (which is allowed since they are degree-2 vertices in $G$).


\item $G$ has a degree-4 vertex $z$ adjacent to two degree-2 vertices $x$ and $y$ such that $N(z) \setminus \{x,y\} = \{N(x) \cup N(y) \} \setminus \{z\}$ (as depicted in Figure \ref{Fig_Case3}): \\
Let $N(z) \setminus \{x,y\} = \{N(x) \cup N(y) \} \setminus \{z\} = \{x', y'\}$, where $x' \in N(x) \setminus \{z\}$ and $y' \in N(y) \setminus \{z\}$.

We now construct a simple and biconnected $SP$ graph $G'$ with $n-2$ vertices from $G$ by suppressing vertices $x$ and $y$ and deleting all but one copy of the resulting parallel edges (cf. Figure \ref{Fig_Case3}). Note that $z$ is a degree-2 vertex in $G'$ and $x',y'$ may also be of degree $2$ in $G'$. By the inductive hypothesis (as $G'$ is a simple and biconnected graph with $5 \leq n-2 <n$ vertices), there exists a valid spanning tree $T'$ for $G'$. As $z, x'$ and $y'$ are potentially all degree-2 vertices in $G$, they are all potentially leaves in $T'$. However, they cannot simultaneously be leaves, because then $T'$ would not be connected. Thus, we distinguish between different cases:
	\begin{itemize}
	\item $x'$ and $y'$ are leaves in $T'$. This case cannot happen since $T'$ would not be connected. 
	\item $z$ and $y'$ are leaves in $T'$. In this case, we add the edges $\{z,x\}$ and $\{y',y\}$ to $T'$ in order to obtain a valid spanning tree $T$ for $G$ (in which $z$ and $y'$ are internal vertices and $x$ and $y$ are leaves).
	\item $x'$ and $z$ are leaves in $T'$. In this case, we add the edges $\{x',x\}$ and $\{z,y\}$ to $T'$ and obtain a valid spanning tree $T$ for $G$ (in which $z$ and $x'$ are internal vertices and $x$ and $y$ are leaves).
	\item $x'$ is a leaf in $T'$. In this case, we add the edges $\{x',x\}$ and either one of $\{z,y\}$ or $\{y',y\}$ to $T'$ and obtain a valid spanning tree $T$ for $G$ (in which $x'$ is an internal vertex and $x$ and $y$ are leaves).
	\item $y'$ is a leaf in $T'$. In this case, we add the edges $\{y',y\}$ and either one of $\{z,x\}$ or $\{x',x\}$ to $T'$ and obtain a valid spanning tree $T$ for $G$ (in which $y'$ is an internal vertex and $x$ and $y$ are leaves).
	\item $z$ is a leaf in $T'$. In this case, we add the edges $\{z,x\}$ and $\{y',y\}$ or $\{z,y\}$ and $\{x',x\}$ to $T'$ and obtain a valid spanning tree $T$ for $G$ (in which $z$ is an internal vertex and $x$ and $y$ are leaves).
	\item Neither $x', y'$ or $z$ are leaves in $T'$. In this case, we can for example add the edges $\{x',x\}$ and $\{y',y\}$ to $T'$ in order to obtain a valid spanning tree $T$ for $G$.
	\end{itemize}

\begin{figure}[htbp]
\centering
\includegraphics[scale=0.3]{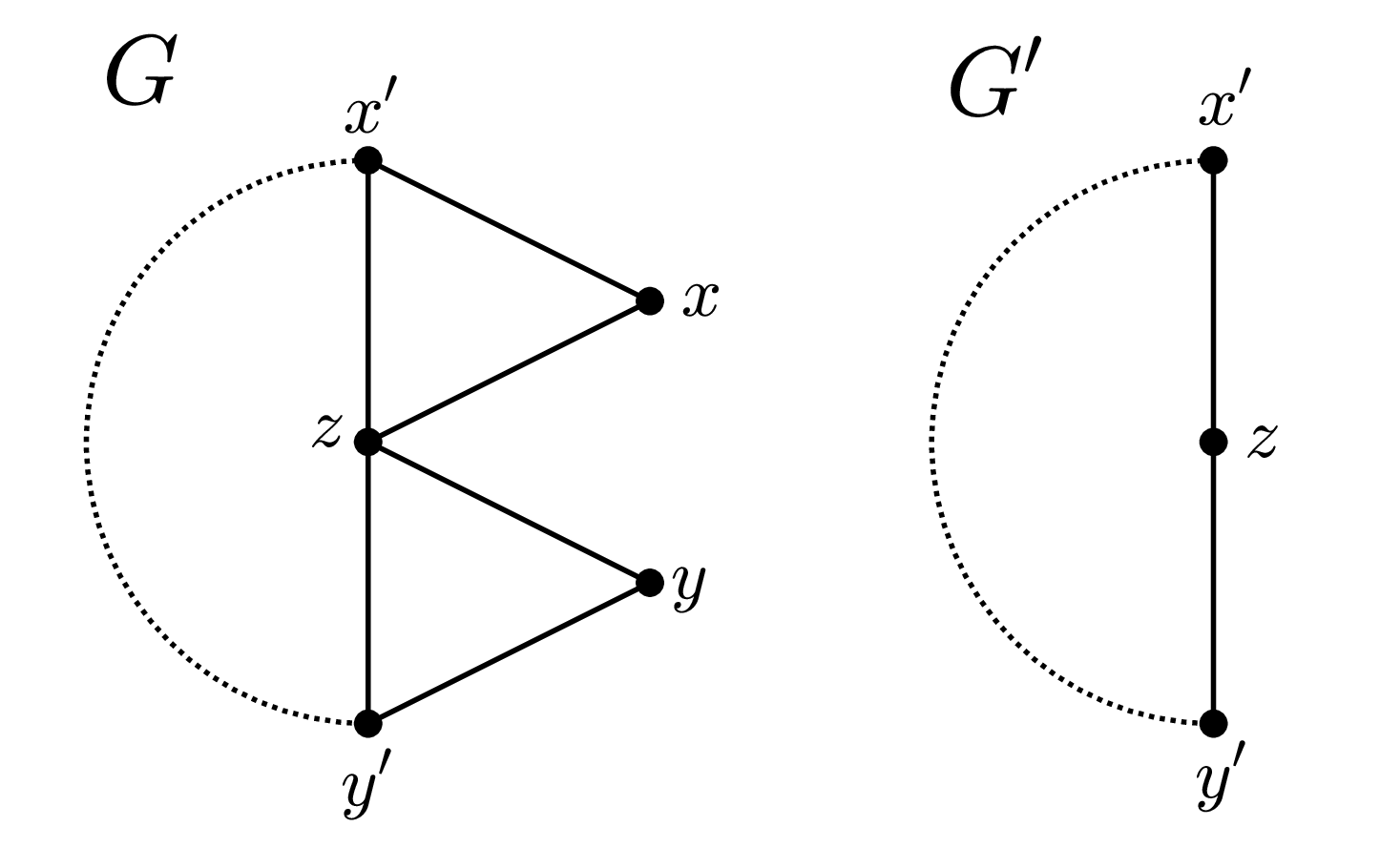}
\caption{Case 3 in the proof of Lemma \ref{Lemma_SP_SpanningTrees}. The dotted line depicts some path between $x'$ and $y'$ which must exist since $G$ is biconnected and $n \geq 7$.}
\label{Fig_Case3}
\end{figure}


\item $G$ has a degree-3 vertex $w$ with $N(w)=\{x,y,z\}$ such that both $x$ and $y$ are degree-2 vertices, $N(x)=\{z,w\}$ and edge $\{y,z\} \notin E$ (as depicted in Figure \ref{Fig_Case4}): \\
Let $y' \neq w$ be the second vertex adjacent to $y$. Note that we cannot have $y'=z$, because $\{y,z\} \notin E$.

We now construct a simple and biconnected $SP$ graph $G'$ with $n-1$ vertices from $G$ by suppressing vertex $y$ (cf. Figure \ref{Fig_Case4}). By the inductive hypothesis (as $G'$ is a simple an biconnected $SP$ graph with $6 \leq n-1 < n$ vertices) $G'$ has a valid spanning tree $T'$ and we can obtain a valid spanning tree $T$ for $G$ from $T'$ as follows:
	\begin{itemize}
	\item If edge $\{w,y'\} \in E(T')$, we replace this edge by $\{w,y\}$ and $\{y,y'\}$ to obtain $T$. 
	\item If edge $\{w,y'\} \notin E(T')$, we add either one of the edges $\{w,y\}$ or $\{y',y\}$ to $T'$ in order to obtain $T$, i.e. we add $y$ as a leaf to $T$ (which is allowed since $y$ has degree 2 in $G$).
	\end{itemize}

\begin{figure}[htbp]
\centering
\includegraphics[scale=0.3]{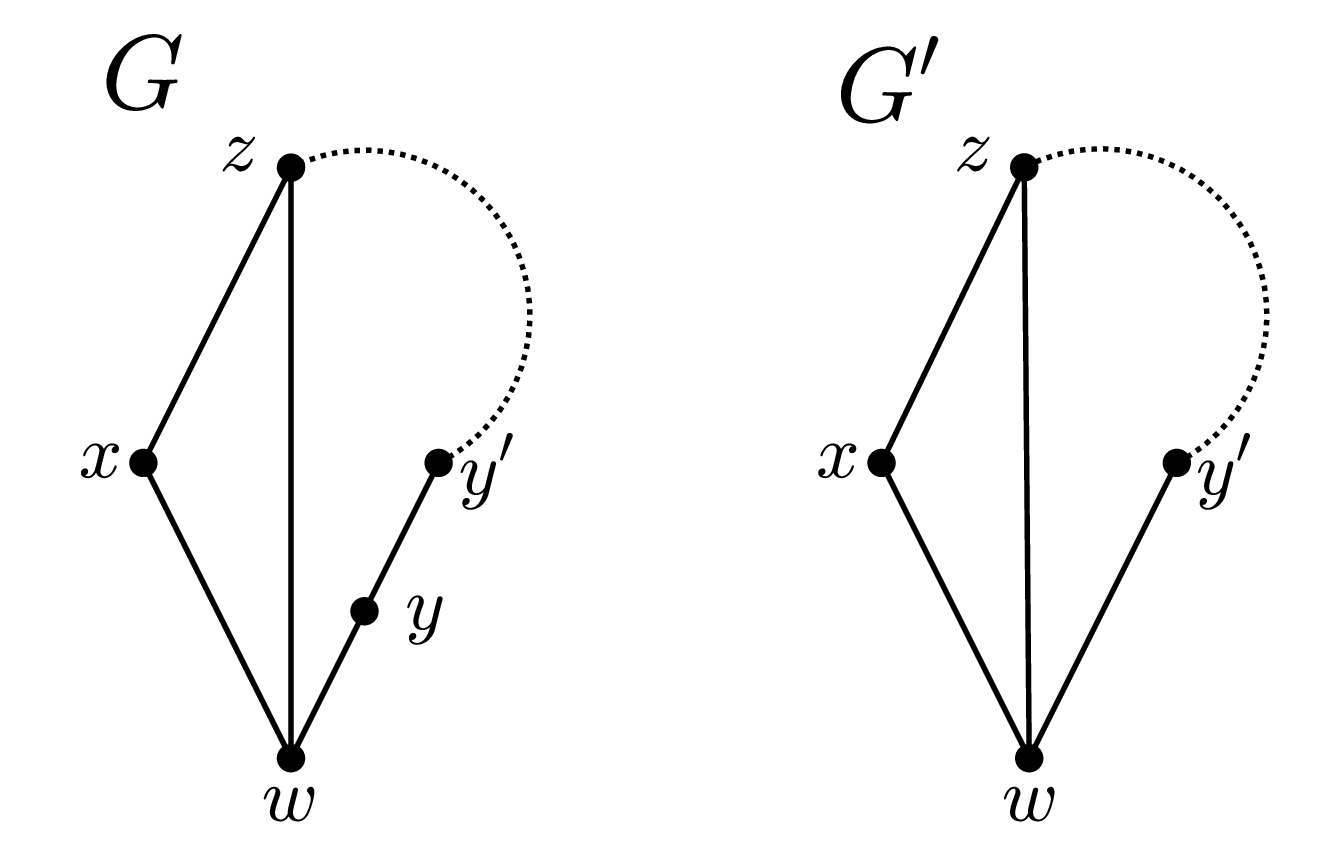}
\caption{Case 4 in the proof of Lemma \ref{Lemma_SP_SpanningTrees}. The dotted line depicts some path between $z$ and $y'$ which must exist since $G$ is biconnected and $n \geq 7$.}
\label{Fig_Case4}
\end{figure}


\item $G$ has two adjacent degree-3 vertices $x$ and $y$ such that $N(x) \cap N(y) = \{z\}$ and $N(z) = \{x,y\}$ (as depicted in Figure \ref{Fig_Case5}):\\
Let $x'$ be the vertex in $N(x) \setminus \{y,z\}$ and let $y'$ be the vertex in $N(y) \setminus \{x,z\}$. 

We now construct a simple and biconnected $SP$ graph $G'$ with $n-2$ vertices from $G$ as follows (cf. Figure \ref{Fig_Case5}):
	\begin{itemize}
	\item Suppress the degree-2 vertex $z$;
	\item Delete one copy of the resulting parallel edge $\{x,y\}$;
	\item Suppress the resulting degree-2 vertex $x$.
	\end{itemize}
	
As $G'$ is a simple and biconnected graph with $5 \leq n-2 < n$ vertices, by the inductive hypothesis $G'$ has a valid spanning tree $T'$. Note that as $y$ is a degree-2 vertex in $G'$ it might be a leaf in $T'$. We now construct a valid spanning tree $T$ for $G$ from $T'$ by distinguishing between two cases:
	\begin{itemize}
	\item If edge $\{x',y\} \in E(T)'$ (note that $\{x',y\} \in E(G') \setminus E(G)$), we replace $\{x',y\}$ by $\{x',x\}$ and $\{x,y\}$ and additionally add the edge $\{y,z\}$ to obtain $T$. This guarantees that the degree-3 vertices $x$ and $y$ in $G$ are not a leaves in $T$ and thus $T$ is a valid spanning tree for $G$ (note that $z$ is a leaf in $T$, but as $deg(z)=2$ in $G$, this is valid).
	\item If edge $\{x',y\} \notin E(T')$, we add the edges $\{y,x\}$ and $\{x,z\}$ to $T'$ in order to obtain $T$. Again, while $z$ is a leaf in $T$, $x$ and $y$ are not, and thus $T$ is a valid spanning tree for $G$.
	\end{itemize}

\begin{figure}[htbp]
\centering
\includegraphics[scale=0.3]{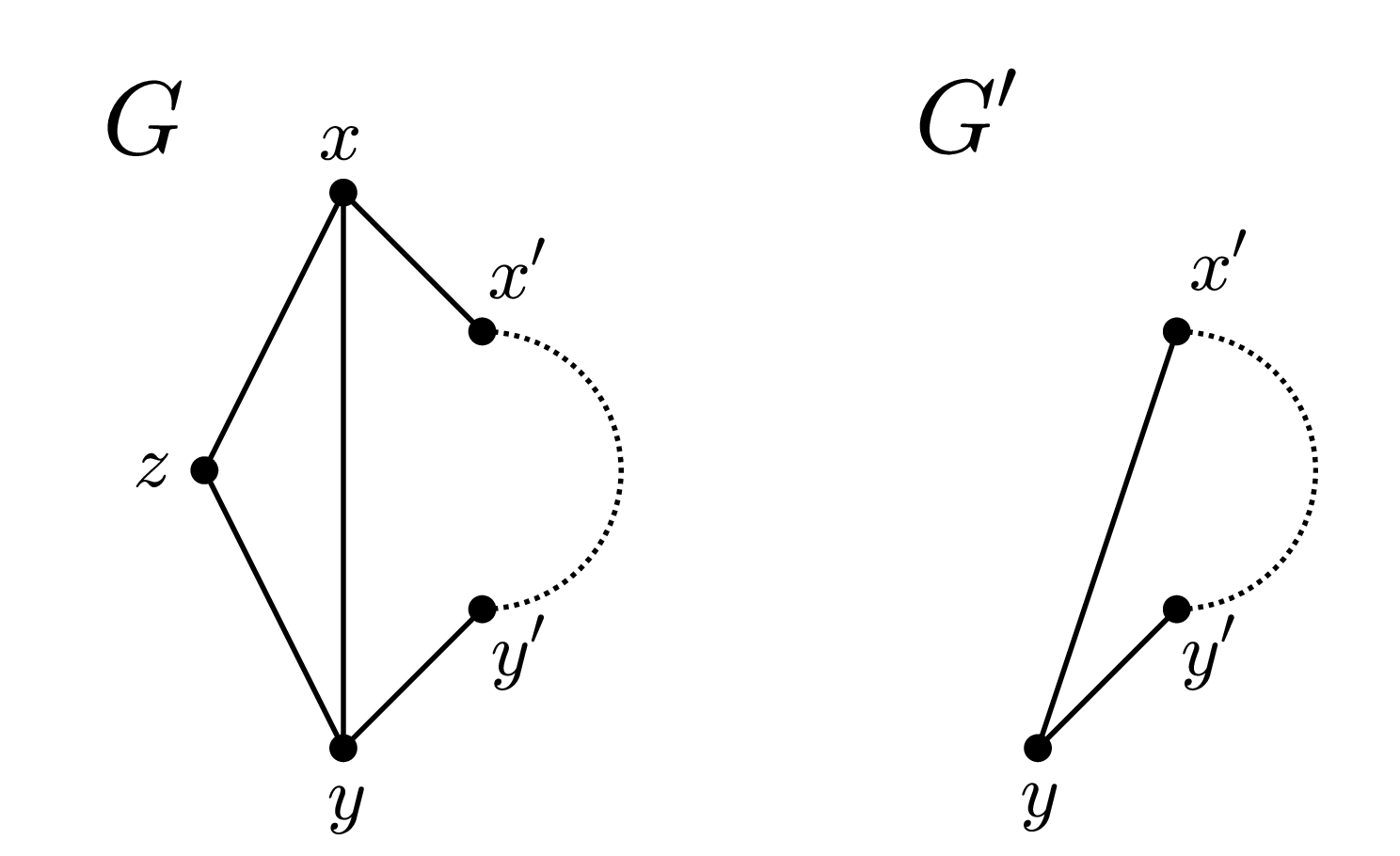}
\caption{Case 5 in the proof of Lemma \ref{Lemma_SP_SpanningTrees} for $n \geq 7$. The dotted line depicts some path between $x'$ and $y'$ which must exist since $G$ is biconnected and $n \geq 7$.}
\label{Fig_Case5}
\end{figure}


\item $G$ has two adjacent degree-3 vertices $w_1$ and $w_2$ such that $N(w_1) = \{x,z_1, w_2\}$, $N(w_2) = \{y, z_2, w_1\}$, $N(x) = \{z_1, w_1\}$ and $N(y) = \{z_2, w_2\}$ (as depicted in Figure \ref{Fig_Case6}): \\
Note that as $n \geq 7$, $z_1$ and $z_2$ are distinct, since otherwise $z_1 = z_2$ would be a cut vertex, contradicting the fact that $G$ is biconnected.
We now construct a simple and biconnected $SP$ graph $G'$ with $n-2$ vertices from $G$ as follows (cf. Figure \ref{Fig_Case6}):
	\begin{itemize}
	\item Suppress the degree-2 vertex $x$ and delete one copy of the resulting parallel edge $\{z_1, w_1\}$;
	\item Suppress the degree-2 vertex $y$ and delete one copy of the resulting parallel edge $\{z_2, w_2\}$.
	\end{itemize}
Note that $w_1$ and $w_2$ are degree-2 vertices in $G'$ and $z_1$ and $z_2$ may be of degree 2 in $G'$ as well. 

By the inductive hypothesis (as $G'$ is a simple and biconnected graph with $5 \leq n-2 < n$ vertices), $G'$ has a valid spanning tree $T'$, in which $w_1, w_2, z_1$ and $z_2$ are potentially leaves. Note, however, that at most two of them can be simultaneously leaves in $T'$, as otherwise $T'$ would be disconnected. We now construct a valid spanning tree $T$ for $G$ from $T'$ by distinguishing between the following cases:
	\begin{itemize}
	\item If $w_1, w_2, z_1$ and $z_2$ are all internal vertices of $T'$, we can for example add the edges $\{w_1, x\}$ and $\{w_2, y\}$ to $T'$ and obtain a valid spanning tree $T$ for $G$.
	\item If $w_1$ is a leaf in $T'$ (and $w_2, z_1, z_2$ are internal vertices in $T'$), we add the edge $\{w_1, x\}$ and either one of the edges $\{w_2, y\}$ or $\{z_2,y\}$ to $T'$ and obtain a valid spanning tree $T$ for $G$.
	\item If $w_2$ is a leaf in $T'$ (and $w_1, z_1, z_2$ are internal vertices in $T'$), we add the edge $\{w_2, y\}$ and either one of the edges $\{w_1, x\}$ or $\{z_1,x\}$ to $T'$ and obtain a valid spanning tree $T$ for $G$.
	\item If $z_1$ is a leaf in $T'$ (and $w_1, w_2, z_2$ are internal vertices in $T'$), we add the edge $\{z_1, x\}$ and either one of the edges $\{w_2, y\}$ or $\{z_2,y\}$ to $T'$ and obtain a valid spanning tree $T$ for $G$.
	\item If $z_2$ is a leaf in $T'$ (and $w_1, w_2, z_1$ are internal vertices in $T'$), we add the edge $\{z_2, y\}$ and either one of the edges $\{w_1, x\}$ or $\{z_1,x\}$ to $T'$ and obtain a valid spanning tree $T$ for $G$.
	\item If $z_1$ and $z_2$ are leaves in $T'$ (and $w_1, w_2$ are internal vertices in $T'$), we add the edges $\{z_1, x\}$ and $\{z_2,y\}$ to $T'$ and obtain a valid spanning tree $T$ for $G$.
	\item If $w_1$ and $w_2$ are leaves in $T'$ (and $z_1, z_2$ are internal vertices in $T'$), we add the edges $\{w_1, x\}$ and $\{w_2,y\}$ to $T'$ and obtain a valid spanning tree $T$ for $G$.
	\item If $w_1$ and $z_1$ are leaves in $T'$ (and $w_2, z_2$ are internal vertices in $T'$), edges $\{w_1, w_2\}$ and $\{w_2, z_2\}$ must be in $T'$ (since $w_2$ is an internal vertex in $T'$). We now remove edge $\{w_1, w_2\}$ from $T'$ (to prevent cycles) and add edges $\{z_1,w_1\}, \, \{w_1, x\}$ as well as $\{w_2, y\}$ to $T'$. This yields a valid spanning tree $T$ for $G$, in which $w_1, w_2, z_1$ and $z_2$ are internal vertices and $x$ and $y$ are leaves.
	\item If $w_1$ and $z_2$ are leaves in $T'$ (and $w_2, z_1$ are internal vertices in $T'$), we add edges $\{w_1, x\}$ and $\{z_2,y\}$ to $T'$ and obtain a valid spanning tree $T$ for $G$.
	\item If $w_2$ and $z_1$ are leaves in $T'$ (and $w_1, z_2$ are internal vertices in $T'$), we add edges $\{z_1, x\}$ and $\{w_2,y\}$ to $T'$ and obtain a valid spanning tree $T$ for $G$.
	\item If $w_2$ and $z_2$ are leaves in $T'$ (and $w_1, z_1$ are internal vertices in $T'$), edges $\{w_1, w_2\}$ and $\{w_1, z_1\}$ must be in $T'$ (since $w_1$ is an internal vertex in $T'$). We now remove edge $\{w_1, w_2\}$ from $T'$ (to prevent cycles) and add edges $\{z_2,w_2\}, \, \{w_2, y\}$ as well as $\{w_1, x\}$ to $T'$. This yields a valid spanning tree $T$ for $G$, in which $w_1, w_2, z_1$ and $z_2$ are internal vertices and $x$ and $y$ are leaves.
	\end{itemize}

\begin{figure}[htbp]
\centering
\includegraphics[scale=0.3]{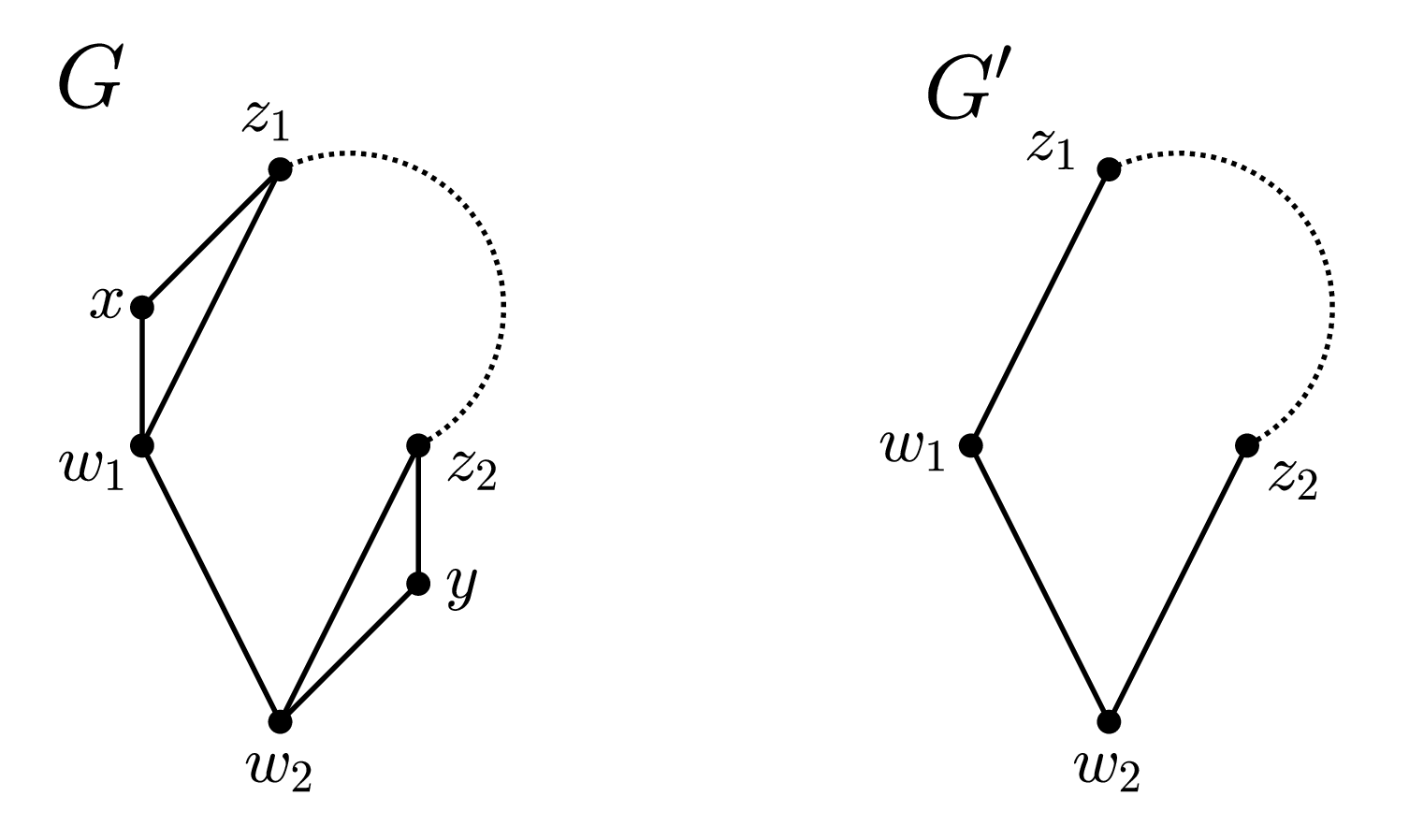}
\caption{Case 6 in the proof of Lemma \ref{Lemma_SP_SpanningTrees} for $n \geq 7$. The dotted line depicts some path between $z_1$ and $z_2$ which must exist since $n \geq 7$ and $G$ is biconnected.}
\label{Fig_Case6}
\end{figure} 


\item $G$ has a degree-3 vertex $w$ with $N(w)=\{x,y,z\}$ such that $N(z)=\{w,y\}$ and edge $\{x,y\} \in E$ (as depicted in Figure \ref{Fig_Case7}): \\
As $n \geq  7$ and $G$ is biconnected, there exists a vertex $u$ in $N(y) \setminus \{w,x,z\}$ (and since $G$ is biconnected $u$ lies on some path from $y$ to $x$). In particular, $deg(y) \geq 4$ in $G$. We now construct a simple and biconnected $SP$ graph $G'$ with $n-1$ vertices from $G$ by suppressing $z$ and deleting one copy of the resulting parallel edge $\{w,y\}$. Note that as $deg(y) \geq 4$ in $G$, we have $deg(y) \geq 3$ in $G'$. In particular, $y$ is not a degree-2 vertex in $G'$, while $w$ is. 
As $G'$ is a simple and biconnected graph with $6 \leq n-1 <n$ vertices, by the inductive hypothesis, $G'$ has a valid spanning tree $T'$, in which vertex $w$ is potentially a leaf. We can now obtain a valid spanning tree $T$ for $G$ from $T'$ by adding the edge $\{w,z\}$ to $T'$. This guarantees that $w$ is not a leaf in $T$ (but $z$ is, which is valid, since $z$ is a degree-2 vertex in $G$) and thus $T$ is a valid spanning tree for $G$.

\begin{figure}[htbp]
\centering
\includegraphics[scale=0.3]{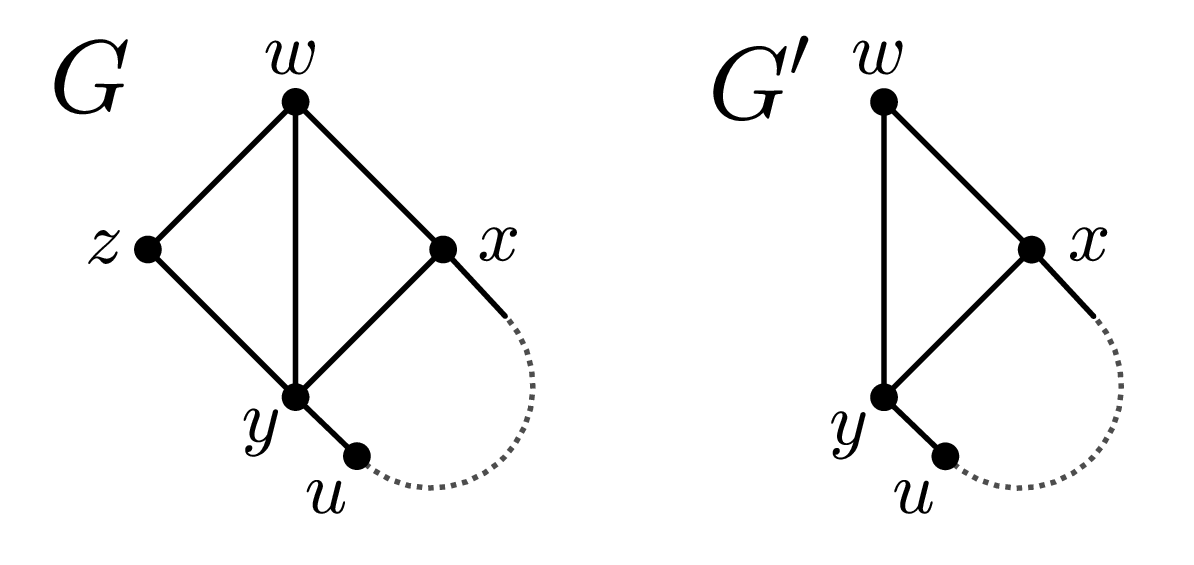}
\caption{Case 7 in the proof of Lemma \ref{Lemma_SP_SpanningTrees}. The dashed line depicts some path between $u$ and $x$ which must exist since $n \geq 7$ and $G$ is biconnected.}
\label{Fig_Case7}
\end{figure}


\item $G$ has two nonadjacent degree-3 vertices $w_1$ and $w_2$ such that $N(w_1) = \{x,y,z_1\}$, $N(w_2) = \{x,y,z_2\}$, $N(z_1)=\{x,w_1\}$ and $N(z_2)=\{y,w_2\}$ (as depicted in Figure \ref{Fig_Case8}): \\
As $G$ is biconnected and $n \geq 7$ there exists a vertex $x' \in N(x) \setminus \{z_1,w_1,w_2,y\}$. In particular, $deg(x) \geq 4$ in $G$.

We now construct a simple and biconnected graph $G'$ with $n-1$ vertices from $G$ by suppressing $z_1$ and deleting one copy of the resulting parallel edge $\{x,w_1\}$. Note that $w_1$ is a degree-2 vertex in $G'$, while $deg(x) \geq 3$ in $G'$. By the inductive hypothesis (as $G'$ is a simple and biconnected $SP$ graph with $6 \leq n-1 < n$ vertices), there exists a valid spanning tree $T'$ for $G'$ (potentially containing vertex $w_1$ as a leaf). We can now obtain a valid spanning tree $T$ for $G$ from $T'$ by adding the edge $\{w_1, z_1\}$. 

\begin{figure}[htbp]
\centering
\includegraphics[scale=0.3]{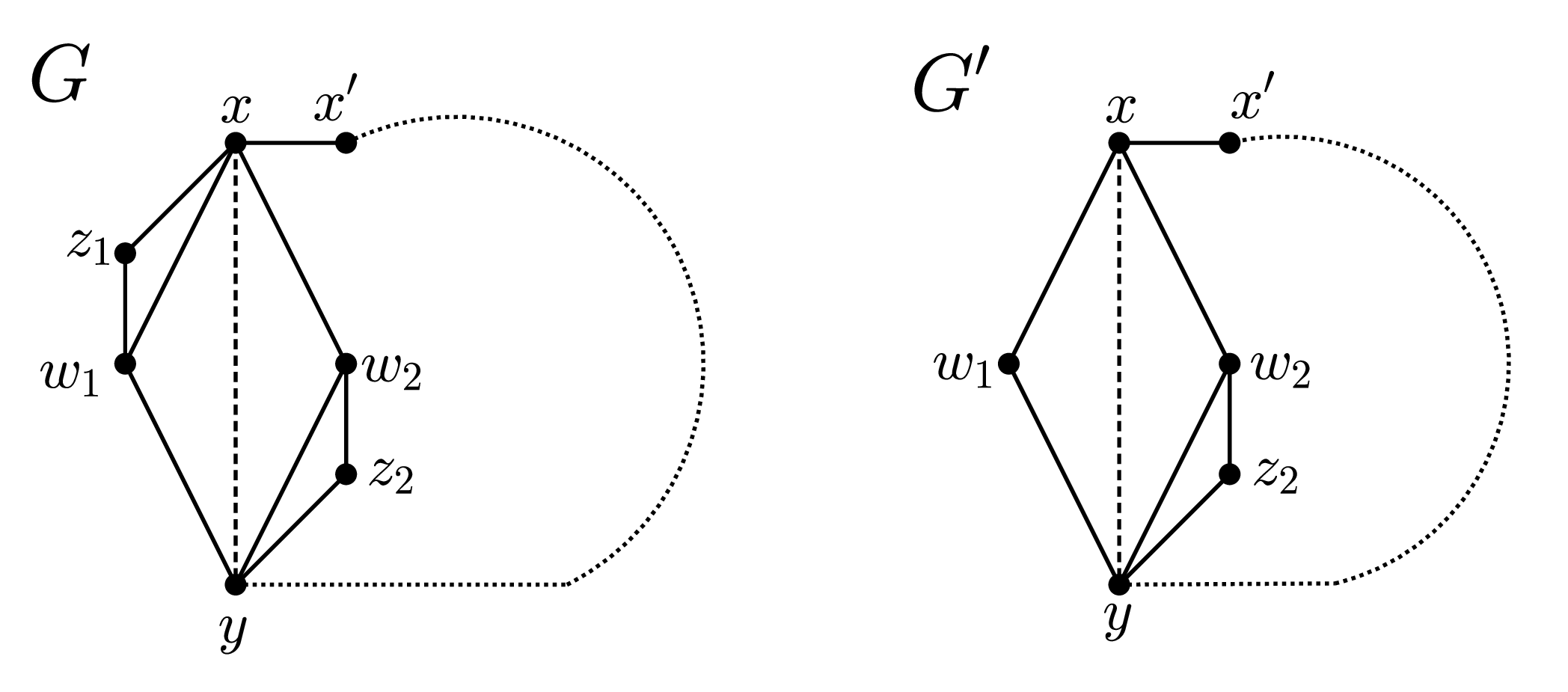}
\caption{Case 8 in the proof of Lemma \ref{Lemma_SP_SpanningTrees} for $n\geq 7$. The dashed edge $\{x,y\}$ may be present or not and the dotted line depicts some path between $x'$ and $y$, which must exist since $G$ is biconnected and $n \geq 7$.}
\label{Fig_Case8}
\end{figure}


\item $G$ has two non-adjacent degree-3 vertices $w_1$ and $w_2$ such that $N(w_1) = \{x,y,z_1\}$, $N(w_2) = \{x,y,z_2\}$, $N(z_1) = \{x,w_1\}$ and $N(z_2) = \{x,w_2\}$ (as depicted in Figure \ref{Fig_Case9}): \\
In this case, we can construct a simple and biconnected $SP$ graph $G'$ with $n-1$ vertices from $G$ by suppressing $z_1$ and deleting one copy of the resulting parallel edge $\{x,w_1\}$ (cf. Figure \ref{Fig_Case9}). Note that this makes $w_1$ a degree-2 vertex in $G'$. By the inductive hypothesis (since $G'$ is a simple and biconnected graph with $6 \leq n-1 <n$ vertices), there exists a valid spanning tree $T'$ for $G'$, in which $w_1$ is potentially a leaf (note that $x$ cannot be a leaf in $T'$ by the inductive hypothesis, since $deg(x) \geq 3$ in $G'$). We can now obtain a valid spanning tree $T$ for $G$ from $T'$ by adding the edge $\{w_1, z_1\}$. This guarantees that $w_1$ is not a leaf in $T$ and thus $T$ is a valid spanning tree for $G$.

\begin{figure}[htbp]
\centering
\includegraphics[scale=0.3]{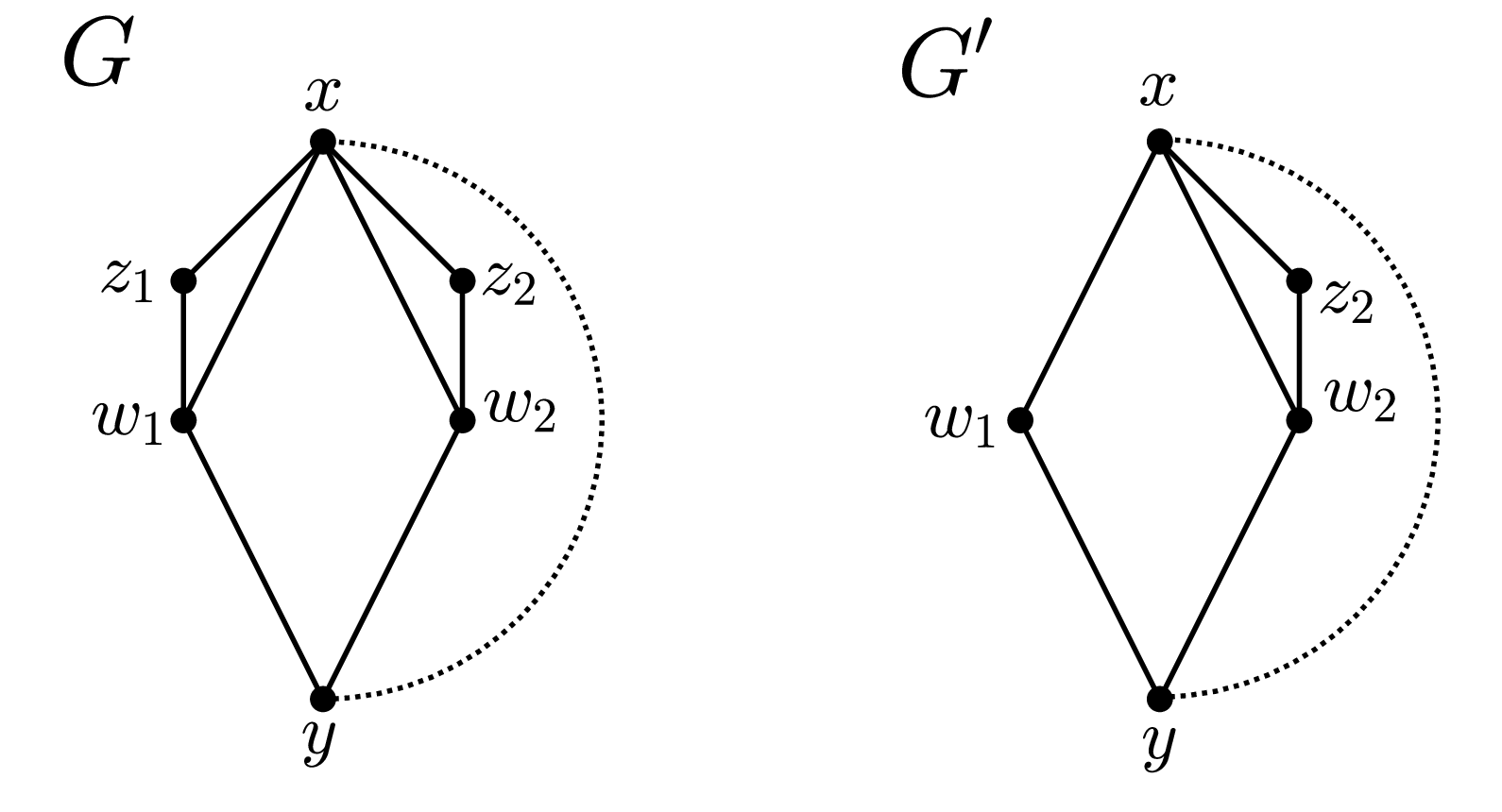}
\caption{Case 9 in the proof of Lemma \ref{Lemma_SP_SpanningTrees}. The dotted line depicts some path between $x$ and $y$ which must exist as $n \geq 7$ (since $G$ is biconnected).}
\label{Fig_Case9}
\end{figure}


\item $G$ has a degree-3 vertex $w$ with $N(w)=\{x,z_1,z_2\}$ such that there is a degree-2 vertex $y \in N(z_1) \cap N(z_2)$ and $N(x) = \{z_1, w\}$ (as depicted in Figure \ref{Fig_Case10}):\\
As $G$ is biconnected and $n \geq 7$ there exists a vertex $z_1' \in N(z_1) \setminus \{x,y,w,z_2\}$ in $G$ and $z_1'$ lies on some path between $z_1$ and $z_2$ (as $G$ is biconnected). In particular, $deg(z_1) \geq 4$ and $deg(z_2) \geq 3$ in $G$.
We now construct a simple and biconnected graph $G'$ with $n-1$ vertices from $G$ by suppressing $x$ and deleting one copy of the resulting parallel edge $\{z_1, w\}$ (cf. Figure \ref{Fig_Case10}). Note that $w$ is a degree-2 vertex in $G'$, while $deg(z_1) \geq 3$ and $deg(z_2) \geq 3$ in $G'$. As $G'$ is a simple and biconnected $SP$ graph with $6 \leq n-1 <n$ vertices, by the inductive hypothesis, there exists a valid spanning tree $T'$ for $G'$, in which $w$ is potentially a leaf. We can now obtain a valid spanning tree $T$ for $G$ from $T'$ by adding the edge $\{w,x\}$, which makes $w$ an internal vertex of $T$ and $x$ a leaf. \qed

\begin{figure}[htbp]
\centering
\includegraphics[scale=0.3]{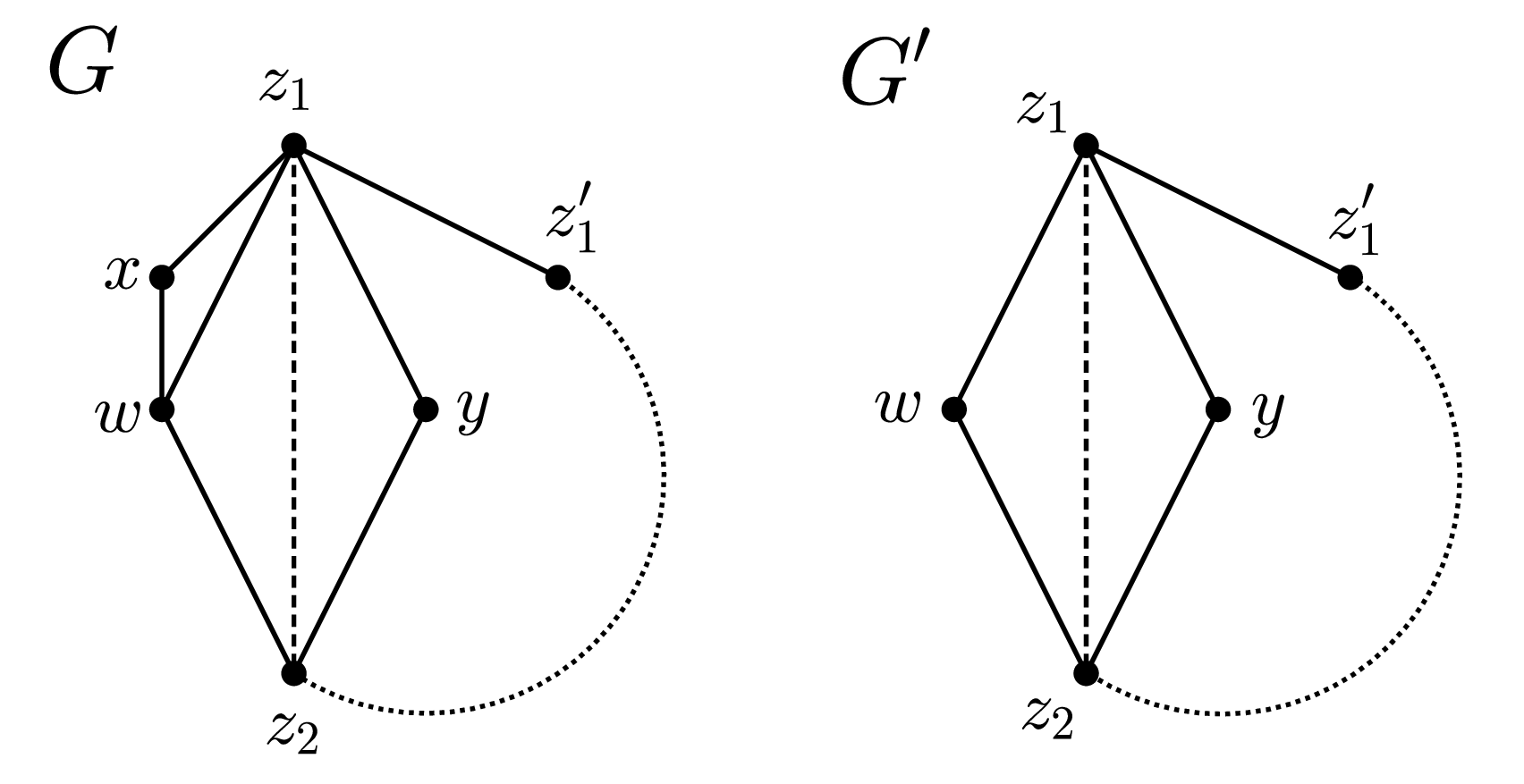}
\caption{Case 10 in the proof of Lemma \ref{Lemma_SP_SpanningTrees} for $n \geq 7$. The dashed edge $\{z_1, z_2\}$ may either be present or absent and the dotted line depicts some path between $z_1'$ and $z_2$ which must exist since $G$ is biconnected.}
\label{Fig_Case10}
\end{figure}
\end{enumerate}
\end{proof}

\begin{figure}[htbp]
\centering
\includegraphics[scale=0.65]{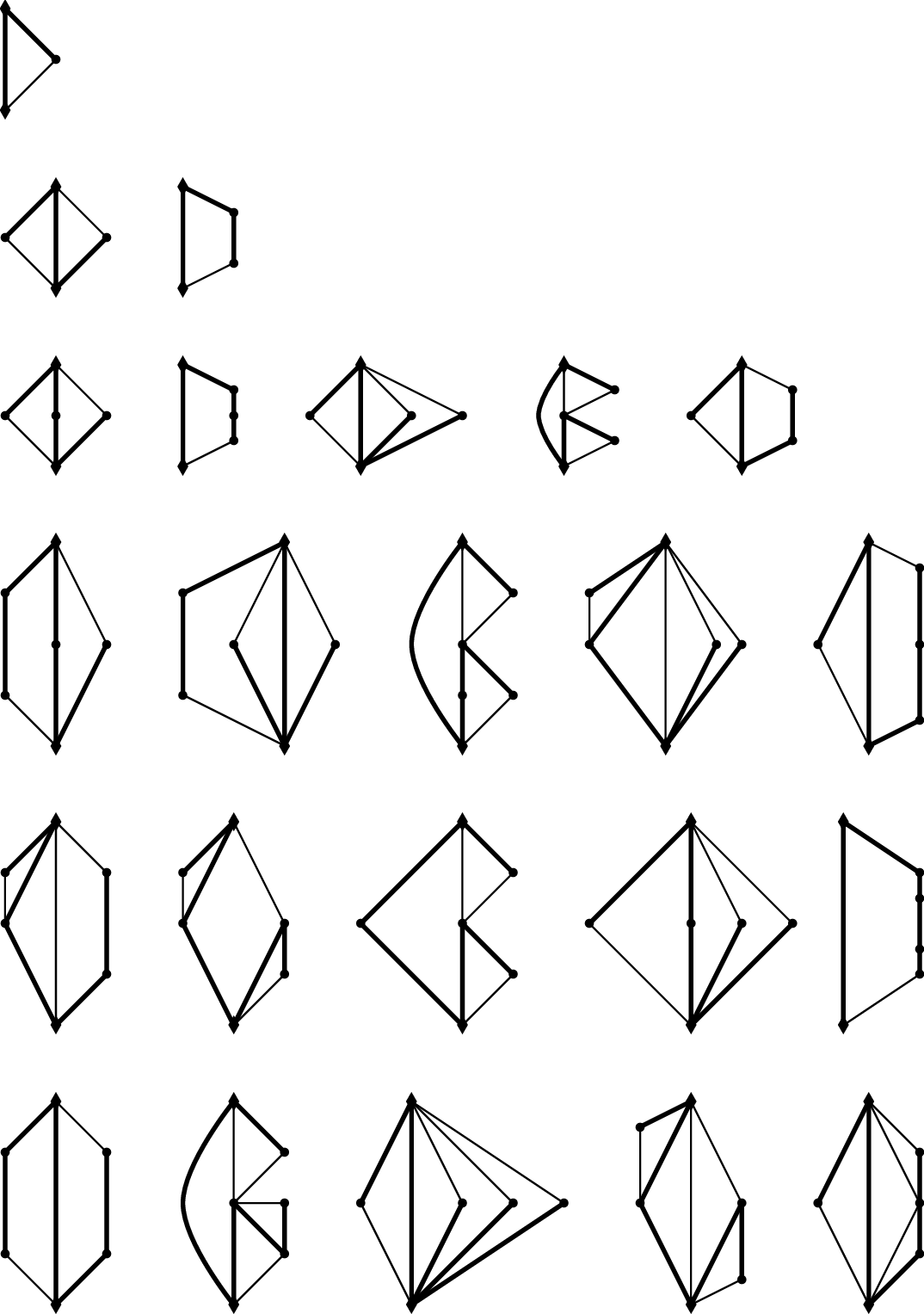} \\
\caption{Catalog of all simple and biconnected $SP$ graphs on $n \leq 6$ vertices for (diamond vertices present one possible pair of terminals, respectively) and a valid spanning tree (depicted in bold).}
\label{Fig_Catalog}
\end{figure}

\setcounter{lemma}{8}
\begin{lemma}\label{chordaltriangle}
Let $G=(V,E)$ be a simple chordal graph without cut edges and with $deg(v)\geq 2$ for all $v \in V$. Then, for every vertex $v \in V$, there exist two other vertices $u$ and $w$, such that $u$, $v$ and $w$ form a triangle in $G$, i.e. such that the edges $\{u,v\}$, $\{u,w\}$ and $\{v,w\}$ are all in $E$.
\end{lemma}

\begin{proof}
Let $G$ be a simple chordal graph without cut edges and with  $deg(v)\geq 2$ for all $v \in V$. 

First we show that every vertex belongs to a cycle. Assume there is a vertex $v$ in $V$ which does not belong to any cycle. As $deg(v)\geq 2$, $v$ has at least two neighbors $a$ and $b$. Now if we remove the edge $e=\{a,v\}$, the resulting graph must still be connected, because otherwise, $e$ would be a cut edge, but $G$ has no cut edge. However, this must mean that there is a path $P$ from $a$ to $v$, which does not use edge $e$. Re-introducing edge $e$ therefore closes a cycle. So indeed, $v$ belongs to a cycle in $G$.

Now assume that $v$ does not belong to a triangle. Then, $v$ belongs to a cycle of length at least four. As $G$ is chordal, this cycle must have a chord. So $v$ belongs also to a smaller cycle. Recursively, this shows that $v$ must belong to a triangle, as all cycles of length larger than three by definition of chordality have a chord. This completes the proof.
\end{proof}

\end{document}